\documentclass[twocolumn,3p]{elsarticle}  

\usepackage{fancyhdr}
 \usepackage{amsmath}
\usepackage{cases}
\usepackage{stfloats}
\usepackage{enumerate}
\usepackage{url}
\usepackage{multicol}
\usepackage{graphicx}
\usepackage{wrapfig}
\usepackage{picinpar}
\usepackage{cutwin}
\usepackage{picins}
\usepackage{balance}
\usepackage{algorithm}
\usepackage{algpseudocode}

 \usepackage{etoolbox}

\newtoggle{ACM}
\toggletrue{ACM}
\togglefalse{ACM}

\newtoggle{IEEEcls}
\toggletrue{IEEEcls}
\togglefalse{IEEEcls}

\newtoggle{ElsJ}
\toggletrue{ElsJ}

\floatname{algorithm}{Algorithm}

\usepackage{amsthm}
\theoremstyle{plain}
\newtheorem{theorem}{Theorem}
\newtheorem{corollary}{Corollary}
\theoremstyle{definition}
\newtheorem{definition}{Definition}

\newif\ifNotUse  
 \NotUsetrue
 
 \newif\ifNotUse  
 \NotUsetrue

\begin{document}
\title{\textit{VATE}: a trade-off between memory and preserving time for high accuracy cardinalities estimation under sliding time window} 

\author[seu_cs]{Jie Xu\corref{cor1}}
\ead{xujieip@163.com}
\author [seu_ns] {Wei Ding}
\ead{wding@carnation.njnet.edu.cn}
\author [seu_ns] {Xiaoyan Hu}
\ead{xyhu@carnation.njnet.edu.cn}

\cortext[cor1]{Corresponding author}
\address[seu_cs]{School of Computer Science and Engineering, South East University, Nanjing, China}
\address[seu_ns]{School of Cyber Science and Engineering, South East University, Nanjing, China}

\begin{abstract}
Host cardinality is one of the important attributes in the field of network research. The cardinality estimation under sliding time window has become a research hotspot in recent years because of its high accuracy and small delay. This kind of algorithms preserve the time information of sliding time window by introducing more powerful counters. The more counters used in these algorithms, the higher the estimation accuracy of these algorithms. However, the available number of sliding counters is limited due to their large memory footprint or long state-maintenance time. To solve this problem, a new sliding counter, asynchronous timestamp (\textit{AT}), is designed in this paper which has the advantages of less memory consumption and low state-maintenance time. \textit{AT} can replace counters in existing algorithms. On the same device, more \textit{AT} can be used to achieve higher accuracy. Based on \textit{AT}, this paper designs a new multi-hosts cardinalities estimation algorithm \textit{VATE}. \textit{VATE} is also a parallel algorithm that can be deployed on GPU. With the parallel processing capability of GPU, \textit{VATE} can estimate cardinalities of hosts in a 40 Gb/s high-speed network in real time at the time granularity of 1 second.
\end{abstract} 
 \maketitle
 
\section{Introduction}
Measuring the attributes of the core network, such as traffic size, packet number, host cardinality and so on, is the basis of network management and research. This paper mainly studies how to estimate the cardinality of host under sliding time window. Cardinality refers to the number of distinct elements in a data stream over a period of time. It is an important attribute in network management and research, and plays an important role in many network applications, such as DDoS attack detection\cite{DosC:ACooperativeIntrusionDetectionFrameworkCloud}\cite{DosC:AnalysisSimulationDDOSAttackCloud} and network scanning\cite{Scan:EvasionResistantNetworkScanDetection}. In the network domain, cardinality can be flow cardinality (the number of distinct flows in a period\cite{HSD:sampleFlowDistributionEstimate}), host cardinality (the number of other hosts communicating with it\cite{HSD:streamingAlgorithmFastDetectionSuperspreaders}\cite{HSD:AcontinuousVirtualVectorBasedAlgorithmMeasuringCardinalityDistribution}). The cardinality estimation algorithm can be applied to all these problems. In this paper, the estimation of host cardinality is taken as the research object, and the algorithm studied in this paper can also be applied to the estimation of flow cardinality.

Suppose there is a core network \textit{ANet}, which is managed by some organisation, institute or ISP (Internet Service Provider). \textit{ANet} communicates with other networks through a group of edge routers \textit{ER}. Use \textit{BNet} to represent other networks that communicate with \textit{ANet}. For a host $aip$ in \textit{ANet}, its cardinality refers to the number of hosts in \textit{BNet} that communicate with it through \textit{ER} in a time window. We call these hosts that communicate with $aip$ in \textit{BNet} as the opposite hosts of $aip$. They communicate with $aip$ in a certain time window, that is, they send packets to $aip$ or receive packets from $aip$. In this paper, the task of host cardinality estimation is to estimate the cardinality of each $aip$ in \textit{ANet} by scanning all packets passing through the \textit{ER}.

The time window can be a discrete time window or a sliding time window, as shown in Figure \ref{fig_sliding_and_discrete_time_window}. Divide the network traffic by successive time slices. The size of the time slice can be set to 1 second, 1 minute, 5 minutes or other lengths suitable for different applications.

\begin{figure}[!ht]
\centering
\includegraphics[width=0.47\textwidth]{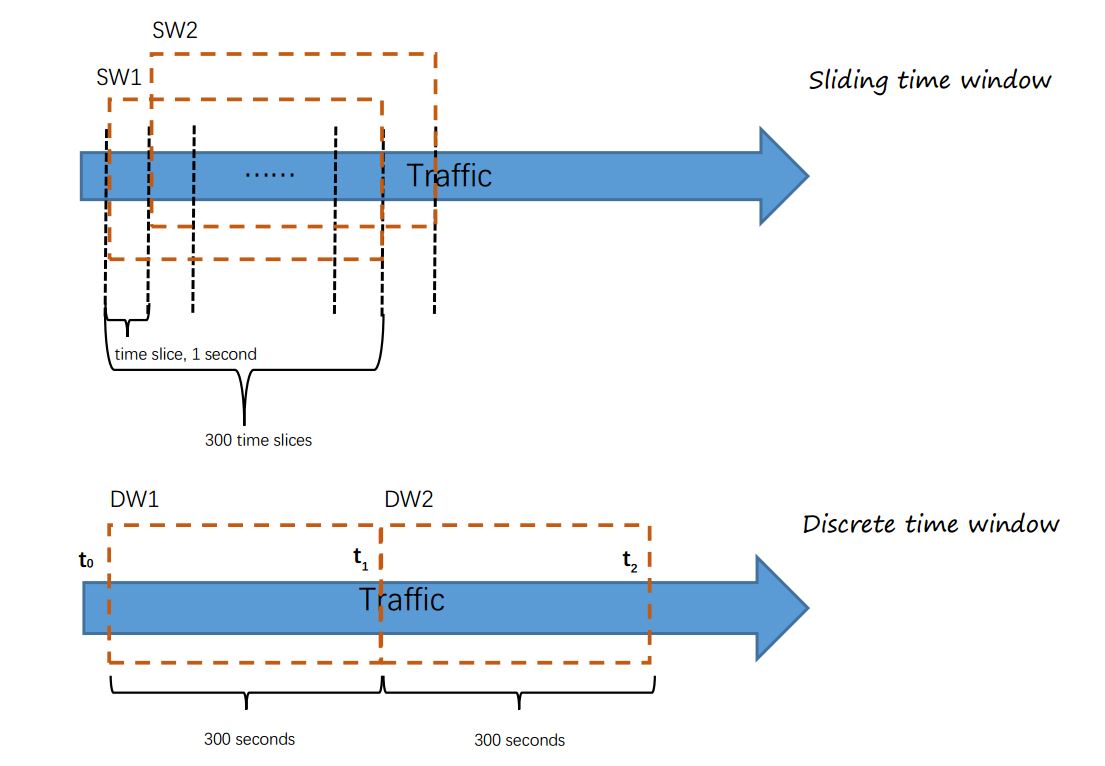}
\caption{Sliding time window and discrete time window}
\label{fig_sliding_and_discrete_time_window}
\end{figure}

The time window moves forward one time slice at a time. A sliding time window can contain up to $k$ consecutive time slices\cite{SDC:IMC2003:IdentifyingFrequentItemsSlidingWindowsOnlinePacketStreams} while a discrete time window has only one time slice. Let $W(t, k')$ denote a time window consisting of consecutive $k'$ time slices starting from time slice $t$, where $k'$ is a positive integer less than or equal to $k$. The size of the time window is the length of time it covers. For a sliding time window containing $k'$ time slices, its size is the sum of the lengths of $k'$ time slices. For discrete time windows, the size is the length of a time slice. It is assumed that the cardinality estimation is performed only at the end of each time slice, that is, at the end of the time slice $t+k'-1$, the cardinalities of the hosts in $W(t, k')$ is estimated. 

Finer-grained time slice makes the cardinality obtained under sliding time window more accurate and timely than that obtained under discrete time window.
But the calculation of cardinality under sliding time window is much more complicated than that under discrete time window, because when sliding time window moves forward, it needs to keep the state of opposite hosts in the previous time slices. The state of opposite host is defined as follows:
\begin{definition}[Active opposite hosts]
\label{def-activeOppositeHost}
Let $OP(aip, t, k')$ represents $aip$'s set of opposite hosts in the time window $W(t, k')$. If $bip$ belongs to $OP(aip, t, k')$, it is said that $bip$ is active for $aip$ in time window $W(t, k')$.
\end{definition}

The cardinality estimation under sliding time window must preserve the active opposite hosts in the previous $k-1$ time slices.
In the cardinal estimation algorithm, $aip$ may contain multiple counters and each counter may correspond to various opposite hosts. The state of the counter is determined by its opposite hosts.
\begin{definition}[Active counter]
\label{def-actveCounter}
In the time window $W(t, k')$, when one of the opposite hosts corresponding to a counter is active, the counter is active; if all the opposite hosts corresponding to it are inactive, the counter is inactive.
\end{definition}

Timestamp (\textit{TS})\cite{SDC:MaintainingStreamStatisticsOverSlidingWindows} is the earliest counter for sliding time window. \textit{TS} recorded the latest appearance time of its corresponding opposite hosts. At a time point, the \textit{TS} is inactive if the difference between the current time and \textit{TS} is greater than the size of the time window. The state of \textit{TS} can be obtained at any time. Because the sliding time window keeps moving forward, \textit{TS} must be large enough, such as 32 or 64 bits. High memory requirements limit the number of timestamps.

For each counter, at the end of a time slice, we only want to determine whether it is active or not. In other words, we are interested in whether the latest opposite host appears in the current time window. Since the length of each time slice is invariant and a sliding time window contains $k$ time slices at most, $log_2(k + 1)$ bits are theoretically sufficient to preserve the counter's state. Distance Recorder (\textit{DR}) \cite{ISPA2017:HighSpeedNetworkSuperPointsDetectionBasedSlidingWindowGPU} is proposed under this idea. Each \textit{DR} occupies $ceil(log_2(k + 1))$ bits, where the $ceil(x)$ function represents the smallest integer no smaller than x. Each bit of \textit{DR} is set to 1 at the beginning. The \textit{DR} is set to 0 if the corresponding host of the \textit{DR} appears in the current time slice. Each time the time window slides, the state of \textit{DR} needs to be updated, that is, the value of \textit{DR} plus 1. At the end of time slice $t$, if the value of \textit{DR} is less than $k'$, the \textit{DR} is active in $W(t-k'+1, k')$. Unlike \textit{TS}, \textit{DR} uses only $ceil(log_2(k + 1))$ bits. But the state of \textit{DR} needs to be maintained at the end of each time slice. When the number of \textit{DR}s is large, maintaining the state of each \textit{DR} is a heavy task. The state-maintenance time of \textit{DR} limits the number of \textit{DR}s available.

Both \textit{TS} and \textit{DR} have their strengths and weaknesses. Can we use less memory and less maintenance time to determine the state of the counter? If we want to save memory, we must regularly maintain the state of each counter and mark inactive ones; if we want fewer state-maintenance operations, we can reduce the number of counters whose state need to be maintained in each time slice. Inspired by this idea, we propose a new counter, asynchronous timestamp (\textit{AT}). \textit{AT} is a tradeoff between memory consumption and state-maintenance time. It contains $ceil(log_2(2*k+1))$ bits. Although \textit{AT} has one bit more than \textit{DR}, \textit{AT} only needs to maintain its state once every $k$ time slices. In other words, \textit{AT} reduces the time complexity of state-maintenance to $O(1/k)$ at the cost of an extra bit. Based on \textit{AT}, an efficient parallel algorithm, virtual asynchronous timestamp estimator (\textit{VATE}), is proposed to estimate the cardinalities under sliding time windows. The main contributions of this paper are as follows.

\begin{itemize}
\item A new counter, asynchronous timestamp (\textit{AT}), is proposed. \textit{AT} only occupies $ceil(log_2(2*k+1))$ bits and can maintain its state once every $k$ time slices. \textit{AT} can determine its state at the end of each time slice and is suitable for running in the parallel environment.
\item A new algorithm for estimating the cardinalities under sliding time window is proposed, which is called virtual asynchronous timestamp estimation algorithm \textit{VATE}. \textit{VATE} uses a fixed number of $AT$s to estimate cardinalities. Because a single \textit{AT} takes up less memory and has low state-maintenance time, \textit{VATE} can contain more $AT$s to achieve higher estimation accuracy than other algorithms.
\item The \textit{VATE} algorithm is deployed in GPU to acquire real-time cardinalities estimation in high-speed network under sliding time windows.
\item Using real high-speed network data, we compare the accuracy of \textit{VATE} under different parameters and the running time on different GPU.
\end{itemize}

This article is arranged as follows. In the next section, we will introduce existing works of cardinality estimation under sliding time window. In Section 3, we will describe how \textit{AT} works and why it can use less memory and less maintenance time to determine its own state. Section IV presents an \textit{AT}-based algorithm for estimating cardinalities under sliding time window, \textit{VATE}. It also introduces how to deploy \textit{VATE} on GPU and realize the cardinalities estimation in parallel. In Section 5, the experiments of real high-speed network traffic are shown. In section 6, we summarize this paper.

\section{Background \& Related works}
\subsection{Cardinality estimation}
For a host $aip$ in \textit{ANet}, let $Pkt(aip, t, 1)$ represent the set of all packets that communicate with $aip$ through \textit{ER} in the time window $W(t, 1)$. From each packet in $Pkt(aip, t, 1)$, an IP address pair like $<aip, bip>$ can be extracted, in which $bip$ is the opposite host of aip in the packet. Let $IPair(aip, t, 1)$ represent the stream of IP address pairs extracted from $Pkt(aip, t, 1)$. Because an opposite host $bip$ may send or receive multiple packets in a time window, the same IP address pair may appear many times in $IPair(aip, t, 1)$. The cardinality of aip is calculated by scanning $IPair(aip, t, 1)$ to get the number of distinct IP address pairs, namely $|OP(aip, t, 1)|$.

In recent years, many excellent cardinality estimation algorithms have been proposed. These algorithms\cite{PCSA:ProbabilisticCountingAlgorithmsForDataBaseApplications}\cite{DC:LoglogCountingOfLargeCardinalitiesDurand2003}\cite{DC:HyperLogLogTheAnalysisOfANearoptimalCardinalityEstimationAlgorithm} mostly use a vector containing $g$ counters to estimate the cardinality of a host. These algorithms differ in what is saved in each counter, how to update the counter, and how to estimate the cardinality from the counter vector.

Flajolet et al \cite{PCSA:ProbabilisticCountingAlgorithmsForDataBaseApplications} proposed a cardinality estimation algorithm called PCSA. The counter used by PCSA is a bitmap containing 32 bits. Each host is randomly mapped to a counter, and its lowest significant bit is stored in the counter. When all IP address pairs are scanned, the cardinality is estimated based on the values of all counters.

The task of each counter in PCSA is to record the lowest significant bit of opposite hosts. For IPv4 addresses, the maximum lowest significant bit is 32, which can be expressed in 5 bits. But PCSA uses 32 bits, leaving much room for improvement. For this reason, Philippe et al \cite{DC:LoglogCountingOfLargeCardinalitiesDurand2003} proposed LogLog algorithm. Unlike PCSA, each counter in LogLog records the position of the ``1" bit on the leftmost side of opposite hosts. LogLog estimates the cardinality based on the geometric mean of all counters. Many algorithms are derived from LogLog. Flajolet et al \cite{DC:HyperLogLogTheAnalysisOfANearoptimalCardinalityEstimationAlgorithm} found that when using the harmonic mean of all counters, the accuracy of the estimation results would be improved. Based on this idea, the HyperLogLog algorithm is proposed. MinCount\cite{DC2009:OrderStatisticsEstimatingCardinalitiesMassiveDataSets} is another algorithm similar to LogLog. But it first hashes each host evenly to a real number between [0,1], and each counter stores the minimum hash value it has seen. The size of each counter can be adjusted according to different accuracy.

Although the above algorithms are effective for large cardinality, their accuracy is limited. Whang et al \cite{DC:aLinearTimeProbabilisticCountingDatabaseApp} proposed a high-precision cardinality estimation algorithm based on the maximum likelihood principle, linear estimator LE. A counter in LE has only one bit, and all counters are initialised to zero at the beginning. Each opposite host is mapped to a counter in LE by a random hash function\cite{hash_UniversalClassesOfHashFunctions}. When an opposite host appears, its corresponding counter will be set to 1. At the end of a time window, LE estimates the cardinality of $aip$ according to the following formula, where $g_0$ is the number of `0' bits.
\begin{equation}
\label{eq_linearEstimator}
|OP(aip,t, k)|=-g*ln(\frac{g_0}{g})
\end{equation}

The accuracy of the cardinality estimation algorithm is evaluated by standard error\cite{TON2017_CardinalityEstimationElephantFlowsACompactSolutionBasedVirtualRegisterSharing}. Let n denote the cardinality of $aip$ and n' denote the estimated value obtained by an algorithm. The standard deviation of the estimation algorithm is the standard deviation of n/n', which is recorded as $\sigma$. Table \ref{tbl_cardinalityEstimating_compare} shows the accuracy and memory consumption of different algorithms when $\sigma= 1$ and n = 5000.
\begin{table}
\centering
\caption{Different cardinality estimator compare}
\label{tbl_cardinalityEstimating_compare}
\begin{tabular}{c}                                                                                                                                                                                                                           
\centering
\includegraphics[width=0.45\textwidth]{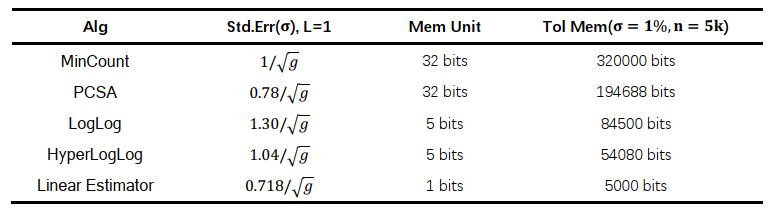}
\end{tabular}
\end{table}

From Table \ref{tbl_cardinalityEstimating_compare}, we can see that LE uses the smallest memory to achieve the same accuracy as other algorithms. In this paper, based on the principle of LE, we design a cardinality estimation algorithm under the sliding time window.

\subsection{Sliding time window and distinct time window}
As shown in Figure \ref{fig_sliding_and_discrete_time_window}, discrete time window and sliding time window are two commonly used time windows in cardinality estimation. The sliding time window moves forward one time slice at a time. Therefore, two adjacent sliding time windows containing $k$ time slices have $k-1$ same time slices. In Figure 1, the size of the time slice is set to 1 second for the sliding time window and 300 seconds for the discrete time window. Sliding time window and discrete time window have the same length, but the sliding time window contains 300 time slices. The granularity of the time slice is smaller than that of the discrete time window.

The cardinality estimation under the discrete time window is relatively simple because it does not need to save the state of the previous time slices. However, the estimation results of cardinalities are influenced by the time window boundary, and the estimation delay is significant. Compared with the discrete time window, the cardinality calculation under the sliding time window has higher accuracy because the sliding time window can monitor network traffic\cite{SDC2010SHLL:SlidingHyperLogLogEstimatingCardinalityDataStreamOverSlidingWindow} on a finer-grained time slice.

There are many algorithms\cite{SDC2007:EstimatingNumberActiveFlowsDataStreamOverSlidingWindow}\cite{SDC2010SHLL:SlidingHyperLogLogEstimatingCardinalityDataStreamOverSlidingWindow}\cite{GLOBECOM2003_CountingNetworkFlowsInRealTime} to solve this problem by using better-functioning counters.

Fusy et al \cite{SDC2007:EstimatingNumberActiveFlowsDataStreamOverSlidingWindow} improved MinCount by using a data structure called List of Future Possible Minima(LFPMin), and proposed a Sliding MinCount. Sliding MinCount replaces each counter in the MinCount with LFPMin so that it can save the smallest packet in the current time window. But the accuracy of Sliding MinCount depends on MinCount, and the length of LFPMin increases with the number of IP address pairs. Inspired by LFPMin, Chabchoub et al \cite{SDC2010SHLL:SlidingHyperLogLogEstimatingCardinalityDataStreamOverSlidingWindow} proposed a data structure called List of Future Possible Maxima(LFPMax) to replace each counter in HyperLogLog, and proposed Sliding HyperLogLog. The accuracy of Sliding HyperLogLog depends on HyperLogLog, and the length of LFPMax increases with the number of IP address pairs. The size of LFPMin and LFPMax is not fixed. When a time window contains $k$ time slices, LFPMin and LFPMax can contain $k$ elements at most or only one element at least. The sizes of LFPMin and LFPMax need to be adjusted dynamically at runtime or reserved the maximum space before running. Dynamically adjusting the size of memory will increase the operating burden of the system, while reserving the maximum required memory space will waste a lot of memory.

Because LE algorithm has high estimation accuracy and simple operation, many algorithms are designed or improved on the basis of LE. Kim et al \cite{GLOBECOM2003_CountingNetworkFlowsInRealTime} replaces bit vectors in LE with time stamp vectors, and proposes a cardinality estimation algorithm under sliding time window based on timestamp vectors, TSV. Each timestamp contains 64 bits. TSV can estimate the cardinality under any time window of any size. But in practice, we don't need to query the cardinality of the host in this way. For windows with $k$ time slices, the size of each counter can be as small as only $log_2(k + 1)$ bits, such as 
\textit{DR} proposed by Jie et al\cite{ISPA2017:HighSpeedNetworkSuperPointsDetectionBasedSlidingWindowGPU}.

A \textit{DR} only occupies $log_2(k + 1)$ bits to record its own state of activity accurately. \textit{DR} can replace each counter in the ACE \cite{Iwqos2017:ACE:PerflowCountingForBigNetworkDataStreamOverSlidingWindows} to estimate cardinalities under sliding time window and we denote this new algorithms as \textit{VDRE}. But when the time window slides, the values of each \textit{DR} need to be updated to maintain their state. When a large number of counters are included, the state-maintenance operation of \textit{DR} takes a lot of time and becomes the bottleneck of real-time running.

In order to improve the speed, Jingsong et al proposed CVS\cite{HSD2016_CVSFastCardinalityEstimationForLargeScaleDataStreamsOverSlidingWindows} algorithm and LRU-Sketch\cite{SDC2017:FastCountingCardinalityOfFlowsBigTrafficOverSlidingWindows} algorithm. CVS adopts the strategy of random updating, and randomly selects a part of the counter to update each time. However, the method of random updating will increase the error of the estimation results. The LRU-Sketch algorithm uses the principle of memory page replacement to delete inactive counters. But the LRU-Sketh algorithm runs in a special sliding time window, that is, there is at most one packet in a time slice.  

In this paper, an algorithm called virtual asynchronous timestamp vector(\textit{VATE}) is proposed to estimate cardinalities under the sliding time window. \textit{VATE} reduces memory occupancy and state-maintenance time at the same time. Table \ref{tbl_slidingCE_compare} compares the differences between these algorithms. All these sliding cardinality estimation algorithms are based on some classical estimation algorithms in the ``basic algorithm" column. ``Memory Unit" column is the number of bits required by each counter. The column ``PT" represents the time complexity of the state-maintenance operation of a counter in each time slice. 
\begin{table}
\centering
\caption{Different sliding time cardinality estimators compare}
\label{tbl_slidingCE_compare}
\begin{tabular}{c}                                                                                                                                                                                                                           
\centering
\includegraphics[width=0.45\textwidth]{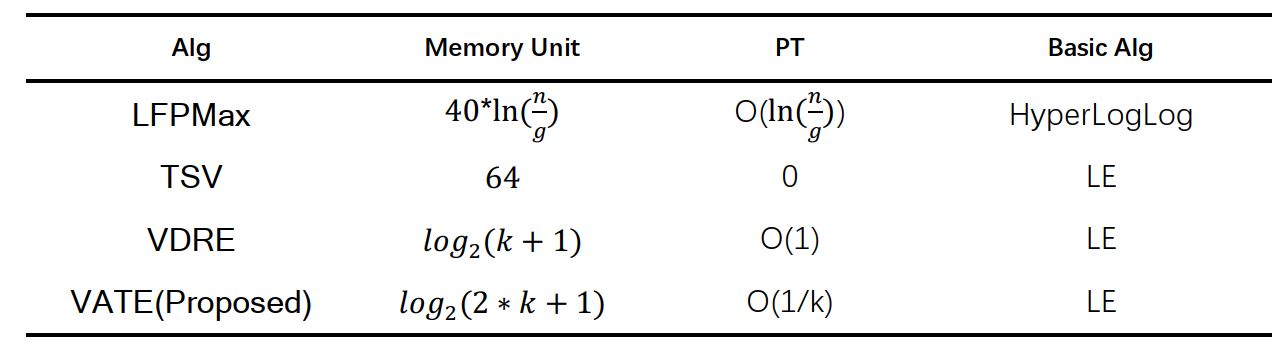}
\end{tabular}
\end{table}

In Table \ref{tbl_slidingCE_compare}, $n$ denotes the number of distinct IP address pairs in a time window, $g$ denotes the number of counters, and $k$ indicates the max number of time slices in a window. Comparing with LFPMax and TSV, \textit{VATE} reduces the memory requirements of each counter. Although the counter used in \textit{VATE} contains one bit more than that in CDV, the state-maintenance time of \textit{VATE} is decreased greatly. \textit{VATE} is also a multi-cardinalities estimation algorithm, that is, it can simultaneously estimate multiple hosts' cardinalities.

\subsection{Multi cardinalities estimation}
In the core network, there are a large number of hosts. The direct way to get all these hosts' cardinalities is to assign an estimator to each of them. But it wastes a lot of memory. Some efficient algorithms use a fixed number of counters to calculate cardinalities of all hosts. These algorithms can be divided into two categories: the algorithm based on estimator matrix and the algorithm based on virtual estimator.

The estimator-matrix based algorithm estimates cardinalities of all hosts using the estimator matrix of the u-row and v-column. For each host, an estimator corresponds to it in each row, i.e. a host cardinality is estimated by using u estimators at the same time. For example, Wang et al designed DCDS\cite{HSD:ADataStreamingMethodMonitorHostConnectionDegreeHighSpeed} based on Chinese remainder theorem; Liu et al designed VBF\cite{HSD:DetectionSuperpointsVectorBloomFilter} based on Bloom filter principle.

The algorithm based on virtual estimator assigns a logical estimator to each host. The counter of each logical estimator is randomly selected from a counter pool, as shown in Figure \ref{fig_virtural_counter_vector_illustrate}.
\begin{figure}[!ht]
\centering
\includegraphics[width=0.47\textwidth]{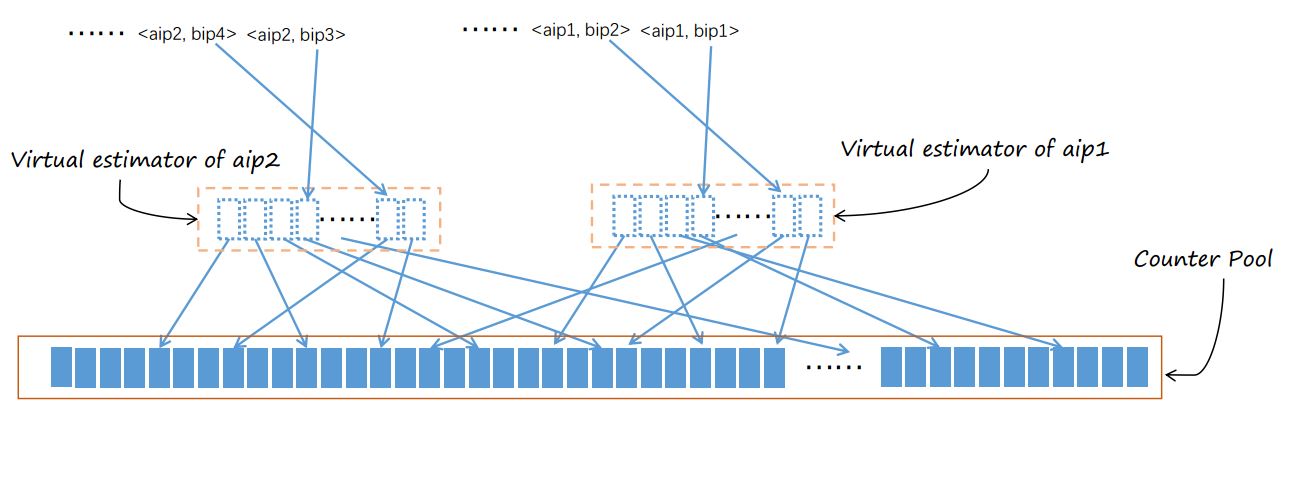}
\caption{Virtual counter vector}
\label{fig_virtural_counter_vector_illustrate}
\end{figure}

The virtual estimator can be LE\cite{HSD:SpreaderClassificationBasedOnOptimalDynamicBitSharing}\cite{HSD:GPU:2014:AGrandSpreadEstimatorUsingGPU}, HyperLogLog \cite{TON2017_CardinalityEstimationElephantFlowsACompactSolutionBasedVirtualRegisterSharing} and so on.
Every counter in the virtual counter vector is relative to a physical counter in the pool. A physical counter is shared by several virtual counter vector. Let $g$ represent the number of counters in a virtual counter vector and $2^c$ represent the number of counters in the pool. A host's cardinality is estimated by the equation \ref{eq_ve_cardinalityEst}.In equation \ref{eq_ve_cardinalityEst}, $Z_v$ represent the fraction of inactive counters in $aip$'s virtual counter vector and $Z_p$ represent the fraction of inactive counters in the counters pool.
\begin{equation}
\label{eq_ve_cardinalityEst}
|OP(aip,t,k)|=g*ln(Z_v)-g*ln(Z_p)
\end{equation}
Equation \ref{eq_ve_cardinalityEst} removes the bias caused by counters sharing. This paper also uses virtual counter vector to estimate host cardinality. But we replace the counter in the pool with \textit{AT} to work under sliding time window.
Unlike estimator-matrix based algorithms, virtual estimator based algorithms only need to update a counter for each packet, because a host only corresponds to an estimator. The \textit{VATE} proposed in this paper is an algorithm based on virtual estimator. \textit{VATE} uses \textit{AT} as the counter of the virtual estimator to estimate multi-host cardinalities under the sliding time window. \textit{AT} and \textit{VATE} are described below in detail.

\section{Asynchronous timestamp}
In a discrete-time window, one bit is sufficient to indicate whether the host has appeared in the current time window. For example, a bit is initialised to zero at the beginning of a time window. It will be set to 1 if one or some of its opposite hosts appear in the time window; otherwise, it will remain zero. But a bit has only two possible values, that is, the counter is active or inactive in the current time slice. Counter having only one bit limits the application of  cardinality estimation under the sliding time window. The critical step in estimating cardinality under sliding time window is to determine whether the counter is active in the current time window(in successive $k'$ time slices). In this section, we will introduce a new counter, Asynchronous Timestamp (\textit{AT}), to solve this problem.

\textit{AT} is a counter consisting of $ceil(log_2(2*k+1))$ bits. So \textit{AT} can represent at least $2*k+1$ different values, from 0 to $2*k$. Each \textit{AT} also has an asynchronous current timestamp (\textit{ACT}). \textit{ACT} is a runtime variable that does not occupy memory space and has a range of values from 0 to $2*k-1$. The \textit{ACT} remains unchanged for a time slice. Let \textbf{\textit{at}} denote an instance of \textit{AT}, $Value(\textbf{\textit{at}})$ denotes the current value of \textbf{\textit{at}}, and $ACT_t^{at}$ denotes the value of the asynchronous current timestamp of \textbf{\textit{at}} in time slice $t$. The relationship between the values of \textit{ACT} in two adjacent time slices conforms to the equation \ref{eq-adjecentTS-ACT}.
\begin{equation}
\label{eq-adjecentTS-ACT}
ACT_{t+1}^{\textbf{\textit{at}}}=(ACT_{t}^{\textbf{\textit{at}}}+1)mod(2*k)
\end{equation}

One of the key tasks of \textit{AT} is to determine whether it is active in the current time window based on its recorded value and \textit{ACT}. Let the value $2*k$ denote the inactive state of \textit{AT}, that is, if $Value(\textbf{\textit{at}})$ equals $2*k$, then \textbf{\textit{at}} is inactive. When $Value(\textbf{\textit{at}})$ is less than $2*k$, its activity needs to be determined by $ACT^{\textbf{\textit{at}}}_t$. \textit{AT} has four operations: initializing($InitAT()$), setting($SetAT()$), maintaining state($PreserveAT()$) and checking state($CheckAT()$). In the following section, t represents the current time slice if there is no special explanation.
\begin{itemize}
\item $InitAT(\textbf{\textit{at}})$: Set the value of \textbf{\textit{at}} to $2*k$. \textit{AT} is set inactive at initialization.
\item $SetAT(\textbf{\textit{at}})$: Set the value of \textbf{\textit{at}} to $ACT^{\textbf{\textit{at}}}_t$.
\item $CheckAT(\textbf{\textit{at}}, k')$: At the end of the time slice $t$, return whether \textbf{\textit{at}} is active in the time window $W(t-k'+1, k')$, and $k'$ is a positive integer not greater than $k$. The specific operation is shown in the algorithm \ref{alg-checkAT_at_k}.
\item $PreserveAT(\textbf{\textit{at}})$: Update the value of \textbf{\textit{at}} at the beginning of some time slices, as shown in the algorithm \ref{alg-preserveAT_at}.
\end{itemize}

\begin{algorithm}         
\caption{checkAT}  
\label{alg-checkAT_at_k}  
\begin{algorithmic}[1]
\Require {
 Asynchronous timestamp $\textbf{\textit{at}}$,
 Time slices number $k'$} 
\Ensure{ActiveState} 
\State $act \Leftarrow ACT_{t}^{at}$
\If{$Value(\textbf{\textit{at}}) == 2*k$}
\State Return False
\EndIf
\State $dis \Leftarrow (act+2*k-Value(\textbf{\textit{at}})) mod (2*k)$ \label{line-disCal-alg-checkAT_at_k}
\If{$dis \leq k'-1$}
\State Return Ture
\Else
\State Return False
\EndIf
\end{algorithmic}
\end{algorithm}

\begin{algorithm}                       
\caption{preserveAT}          
\label{alg-preserveAT_at}                            
\begin{algorithmic}[1]                  
\Require {
 Asynchronous timestamp $\textbf{\textit{at}}$,
 Max time slices number $k$} 
\State $act \Leftarrow  ACT_{t}^{at}$
\State $dis \Leftarrow (act+2*k-Value(\textbf{\textit{at}})) mod (2*k) $
\If{$dis \geq k$}
\State $Value(\textbf{\textit{at}}) \Leftarrow 2*k$
\EndIf
\end{algorithmic}
\end{algorithm}

Firstly, the algorithm \ref{alg-checkAT_at_k} checks whether $Value(\textbf{\textit{at}})$ is $2*k$, and if so, decides directly that \textbf{\textit{at}} is inactive. If not, $ACT_t^{at}$ is used to further determine whether \textbf{\textit{at}} is active. According to $SetAT()$ operation, it is known that \textit{AT} saves its \textit{ACT} of time slice where its opposite hosts latest appearing. Assuming that \textbf{\textit{at}} was last set in the time slice $t'$, and its state has not been maintained by the operation $PreserveAT()$, then at the end of the time slice $t$, $Value(\textbf{\textit{at}}) = ACT^{at}_{t'}$. According to the definition \ref{def-actveCounter}, if $t-t'<k'$, \textbf{\textit{at}} is active in $W(t-k'+1, k')$, otherwise it is inactive. We call $t-t'$ the distance between \textbf{\textit{at}} and the current time slice, abbreviated as $Dist(\textbf{\textit{at}})$. According to the formula \ref{eq-adjecentTS-ACT}, the relationship between $ACT^{at}_t$, $Value(\textbf{\textit{at}})$ and $Dist(\textbf{\textit{at}})$ can be obtained, as shown in the equation \ref{eq-act_value_at_dist_at}.
\begin{equation}
\label{eq-act_value_at_dist_at}
ACT^{at}_t = (Value(\textbf{\textit{at}}) + Dist(\textbf{\textit{at}}))mod (2*k) 
\end{equation}

To reduce memory consumption, \textit{AT} records the value of \textit{ACT} rather than the number of time slices directly. So $Dist(\textbf{\textit{at}})$ can only be calculated from $ACT_t^{at}$. When the value of $Dist(\textbf{\textit{at}})$ does not exceed $2*k$, the calculation formula of $Dist(\textbf{\textit{at}})$, equation \ref{eq-get_dist_at_from_act}, can be obtained from the formula \ref{eq-act_value_at_dist_at}, i.e. the calculation method of line\ref{line-disCal-alg-checkAT_at_k} in the algorithm \ref{alg-checkAT_at_k}.
\begin{equation}
\label{eq-get_dist_at_from_act}
Dist(\textbf{\textit{at}}) =(ACT^{at}_t + 2*k -Value(\textbf{\textit{at}})) mod (2*k)
\end{equation}

According to whether $Dist(\textbf{\textit{at}})$ is less than $k'$, we can get whether \textbf{\textit{at}} is active in the time window $W(t-k'+1, k')$. When $Dist(\textbf{\textit{at}})$ is less than $2*k$, $CheckAT(\textbf{\textit{at}}, k')$ can correctly determine whether \textbf{\textit{at}} is active. However, when $Dist(\textbf{\textit{at}})$ is greater than or equal to $2*k$, $CheckAT(\textbf{\textit{at}}, k')$ may identify inactive \textbf{\textit{at}} as active \textbf{\textit{at}}. For example, when $Value(\textbf{\textit{at}})= ACT^{at}_{t-2*k}$, $Dist(\textbf{\textit{at}}) = 2*k$. However, the value of $Dist(\textbf{\textit{at}})$ calculated from the formula \ref{eq-get_dist_at_from_act} is 0, and $CheckAT(\textbf{\textit{at}}, k')$ will incorrectly determine the state of \textbf{\textit{at}} as active. So the state of \textbf{\textit{at}} must be maintained regularly to ensure that $Dist(\textbf{\textit{at}})$ does not exceed $2*k$ (if $Dist(\textbf{\textit{at}})$ exceeds $2*k$, then \textbf{\textit{at}} is marked inactive, that is, the value of \textbf{\textit{at}} is set to $2*k$). The $PreserveAT()$ operation updates the state of the \textit{AT} at the beginning of a time slice, as shown in the algorithm \ref{alg-preserveAT_at}. The purpose of $PreserveAT()$ operation is to enable $CheckAT(\textbf{\textit{at}}, k')$ operation to accurately determine whether \textbf{\textit{at}} is active. Because $k'$ is not greater than $k$, when $Dist(\textbf{\textit{at}})$ is greater than or equal to $k$, $\textbf{\textit{at}}$ is inactive, that is, when $Dist(\textbf{\textit{at}})$ is greater than or equal to $k$, $PreserveAT(\textbf{\textit{at}})$ marks \textbf{\textit{at}} inactive.

The state of \textit{AT} does not need to be maintained at the beginning of each time slice. The following theorem and corollary show that the state of \textit{AT} can be maintained once every $k$ time slices, and the time complexity of state-maintenance is $O(1/k)$.

\begin{theorem}
\label{thm-psvAT}
Let \textbf{\textit{at}} be an asynchronous timestamp. When the state of \textbf{\textit{at}} is maintained by $PreserveAT()$ operation at the beginning of time slice $t'$, then in time slice $t'$ and the $k-1$ time slices after $t'$, $CheckAT(\textbf{\textit{at}}, k')$ can accurately give whether \textbf{\textit{at}} is active, where $k'\leq k$.
\end{theorem}
\begin{proof}
Let $t''$ denote any time slice between time slice $t'$ and time slice $t'+k-1$. According to the setting operation of \textit{AT} and the algorithm \ref{alg-checkAT_at_k}, if \textbf{\textit{at}} is set in a time slice between $t'$ and $t''$ ($setAT (\textbf{\textit{at}})$), then $CheckAT(\textbf{\textit{at}}, k')$ can output the exact state of \textbf{\textit{at}}.

Next, we discuss the case where \textbf{\textit{at}} is not set between $t'$ and $t''$. Let $v_0$ and $d_0$ denote the value and distance of \textbf{\textit{at}} before it was maintained by $PreserveAT()$ operation at the beginning of $t'$. When the $PreserveAT()$ operation is used to maintain the state of \textbf{\textit{at}}, the value of \textbf{\textit{at}} has two situations: the value of \textbf{\textit{at}} becomes $2*k$(\textbf{\textit{at}} becomes inactive) and the value of \textbf{\textit{at}} is still $v_0$. For the first case, $CheckAT(\textbf{\textit{at}}, k')$ can accurately identify \textbf{\textit{at}} as inactive at the end of time slice $t''$. For the second case, $CheckAT(\textbf{\textit{at}}, k')$ needs to determine whether \textbf{\textit{at}} is active according to the $Dist(\textbf{\textit{at}})$ at the end of $t''$. If the distance at the end of $t''$ does not exceed $2*k$, $CheckAT(\textbf{\textit{at}}, k')$ can give the state of \textbf{\textit{at}} accurately. According to the process of $PreserveAT()$ operation, the value of \textbf{\textit{at}} remains unchanged only when the distance of \textbf{\textit{at}} is less than $k$. So at the end of $t''$, if $Value(\textbf{\textit{at}})$ equals $v_0$, then $d_0 < k$. At the end of $t''$, $Dist(\textbf{\textit{at}}) = d_0 + t''-t'$, and $t''-t'< k' \leq k$. So $Dist(\textbf{\textit{at}})< 2*k$, and $CheckAT(\textbf{\textit{at}}, k')$ can determine whether \textbf{\textit{at}} is active according to whether $Dist(\textbf{\textit{at}})$ is less than $k'$.

In summary, for any of the above cases, $CheckAT(\textbf{\textit{at}}, k')$ can accurately give the state of \textbf{\textit{at}} in the time slice $t'$ and the $k-1$ time slices after $t'$.
\end{proof}

Based on theorem \ref{thm-psvAT}, corollary \ref{cly-preserveAT_at_act_0_k} gives a method to determine when to maintain the state of \textit{AT} according to the value of \textit{ACT}.
\begin{corollary}
\label{cly-preserveAT_at_act_0_k}
Let \textbf{\textit{at}} be an asynchronous timestamp. If $PreserveAT()$ is performed on \textbf{\textit{at}} when $ACT_t^{at}$ is 0 or $k$, at the end of any time slices, $CheckAT(\textbf{\textit{at}}, k')$ can accurately give whether \textbf{\textit{at}} is active or not, where $k' \leq k$.
\end{corollary}
\begin{proof}
Assuming $ACT_{t'}^{at}=0$, according to the formula \ref{eq-adjecentTS-ACT}, we can get $ACT_{t'+k}^{at}=k$ and $ACT_{t'+k+k}^{at}=0$. According to Theorem \ref{thm-psvAT}, if $PreserveAT(\textbf{\textit{at}})$ is executed at the beginning of time slice $t'$, $CheckAT(\textbf{\textit{at}}, k')$ can give the state of \textbf{\textit{at}} accurately between time slice $t'$ and time slice $t'+k-1$; if $PreserveAT(\textbf{\textit{at}})$ is executed at the beginning of time slice $t'+k$, then between time slice $t'+k$ to time slice $t'+k+k-1$, $CheckAT(\textbf{\textit{at}}, k')$ can give the state of \textbf{\textit{at}} accurately. Because $PreserveAT()$ is performed when $ACT^{at}_t$ is 0 or $k$, $CheckAT(\textbf{\textit{at}}, k')$ can accurately give the state of \textbf{\textit{at}} between time slice $t'$ and time slice $t'+k+k-1$. When the time slice $t'+k+k$ is entered, the $PreserveAT()$ operation will maintain the state of \textbf{\textit{at}} again. So at the end of any time slices, $CheckAT(\textbf{\textit{at}}, k')$ can accurately give whether \textbf{\textit{at}} is active.
\end{proof}
 
According to corollary 1, the state-maintenance operation of \textit{AT} can be simplified. When $PreserveAT()$ operation uses the algorithm \ref{alg-preserveAT_at} for \textit{AT} state-maintenance, it is necessary to calculate the distance of \textit{AT}, and then compare it with $k$. According to the formula \ref{eq-get_dist_at_from_act}, it requires one addition, one subtraction and one mod operation to calculate \textit{AT} distance. So the algorithm \ref{alg-preserveAT_at} needs four operations to determine whether the distance of \textit{AT} exceeds $k$: one addition, one subtraction, one mod and one comparison. When corollary 1 is used for state-maintenance, only two or three comparing operations are needed to determine whether the distance of \textit{AT} exceeds $k$. Because when $ACT_t^{at}$ is given, we can know which $Value(\textbf{\textit{at}})$ causes $Dist(\textbf{\textit{at}})$ to be greater than or equal to $k$.

According to the formula \ref{eq-get_dist_at_from_act}, at the beginning of time slice $t$, when $ACT_t^{at}=0$, if $0 \leq Value(\textbf{\textit{at}})\leq k$ and \textbf{\textit{at}} is not set in time slice $t$, then at the end of time slice $t$, $Dist(\textbf{\textit{at}})\geq k$; when $ACT_t^{at}=k$, if $k \leq Value(\textbf{\textit{at}})\leq 2*k-1$ or $Value(\textbf{\textit{at}})=0$, and \textbf{\textit{at}} is not set in time slice $t$, then $Dist(\textbf{\textit{at}})\geq k$ at the end of time slice $t$. So at the beginning of time slice $t$, if $ACT_t^{at}=0$ and $0 \leq Value(\textbf{\textit{at}})\leq k$, then \textbf{\textit{at}} becomes inactive; if $ACT_t^{at}=k$ and $k\leq Value(\textbf{\textit{at}}) \leq 2*k-1$ or $Value(\textbf{\textit{at}})=0$, then \textbf{\textit{at}} becomes inactive.

For example, when $k=9$, Figure \ref{fig_asynTS_updated_everyK} illustrates how to maintain the state of \textit{AT}. In Figure \ref{fig_asynTS_updated_everyK}, the value on the edge of the circle represents the possible values of $Value(at)$. In time slice $t$, if $ACT_t^{at}=17$, then $ACT_{t+1}^{at}=0$. According to corollary \ref{cly-preserveAT_at_act_0_k}, at the beginning of time slice $t+1$, the state of \textbf{\textit{at}} should be maintained. As can be seen from the upper part of Figure \ref{fig_asynTS_updated_everyK}, at the end of time slice $t$, if $Value(\textbf{\textit{at}})$ is between [10,17] and \textbf{\textit{at}} is not set in time slice $t+1$, then $Dist(\textbf{\textit{at}}) < 9$ at the end of time slice $t+1$, and in this case, the value of \textbf{\textit{at}} remains unchanged. But at the end of time slice $t$, if $Value(\textbf{\textit{at}})$ is between $[0,9]$ and \textbf{\textit{at}} is not set in time slice $t+1$, then $Dist(\textbf{\textit{at}}) \geq 9$ at the end of time slice $t+1$, and in this case, \textbf{\textit{at}} should be marked inactive. By analogy, when $ACT_{t+1}^{at}=9$, if $Value(\textbf{\textit{at}})$ is between $[9, 17]$ or $Value(\textbf{\textit{at}})$ is equal to 0, \textbf{\textit{at}} is marked inactive. A slight difference from the case when $ACT_{t+1}^{at}=0$ is that when $ACT_{t+1}^{at}=k$, one more comparing operation (compare whether $Value(\textbf{\textit{at}})$ is 0) is needed to determine whether $Dist(\textbf{\textit{at}})$ is less than $k$.
\begin{figure}[!ht]
\centering
\includegraphics[width=0.47\textwidth]{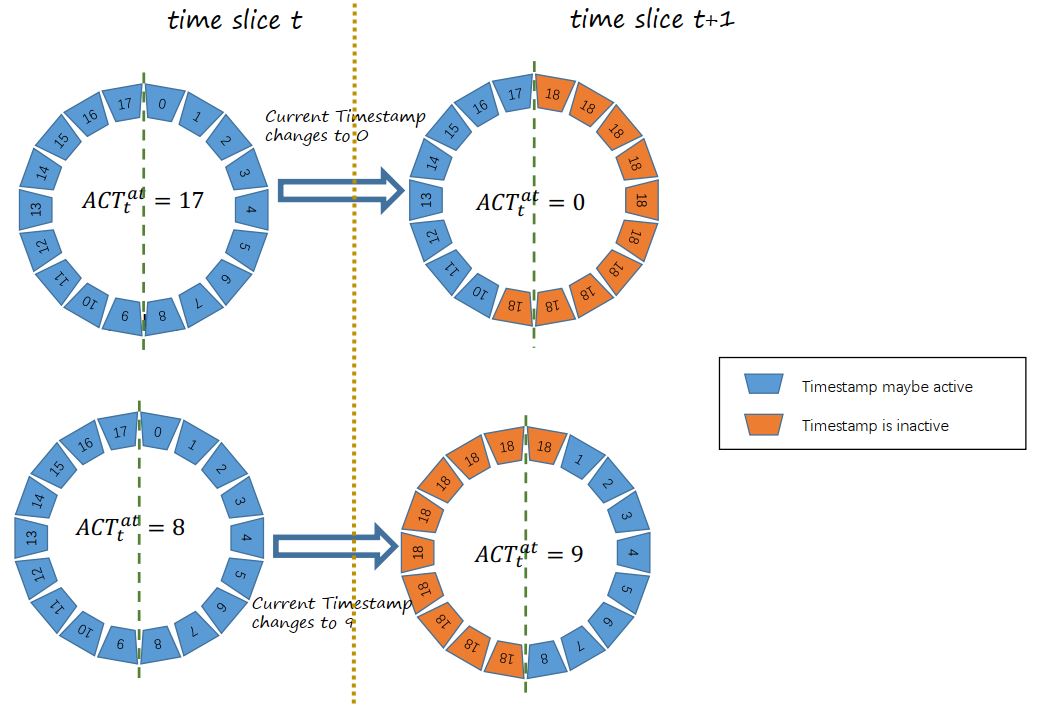}
\caption{Asynchronous timestamp}
\label{fig_asynTS_updated_everyK}
\end{figure}

$PreserveAT()$ operation can use the algorithm \ref{alg-preserveAT_at2} to maintain \textit{AT} state more quickly when following corollary \ref{cly-preserveAT_at_act_0_k}. The algorithm \ref{alg-preserveAT_at2} uses two or three comparisons to determine whether \textit{AT} needs to be marked inactive. The algorithm \ref{alg-preserveAT_at2} is faster than the algorithm \ref{alg-preserveAT_at} because it does not need addition and mod operations.
\begin{algorithm}                       
\caption{preserveAT}          
\label{alg-preserveAT_at2}                            
\begin{algorithmic}[1]                  
\Require {
 Asynchronous timestamp $\textbf{\textit{at}}$,
 Time slices number $k$} 
\State $act \Leftarrow $ the $ACT$ of $at$
\If{$act == 0$}
\If{$0 \leq Value(at)\leq k$}
\State $Value(at) \Leftarrow 2*k$
\EndIf
\EndIf
\If{$act == k$}
\If{$k \leq Value(at) \leq 2*k-1$}
\State $Value(at) \Leftarrow 2*k$
\EndIf
\If{$ Value(at) == 0$}
\State $Value(at) \Leftarrow 2*k$
\EndIf
\EndIf
\end{algorithmic}
\end{algorithm}

In the cardinality estimation algorithm, a large number of fixed counters are used. Because \textit{AT} preserves its state incrementally, when \textit{AT} is used as the counter of these algorithms, cardinality can be estimated under sliding time window continuously. If there are $2*k$ \textit{AT}s and their \textit{ACT}s are different from each other, only 2 \textit{AT}s need $PreserveAT()$ operation at the beginning of each time slice. Based on this idea, we propose a virtual asynchronous timestamp estimator \textit{VATE}. It estimates the cardinality under sliding time window based on a fixed length \textit{AT} pool.

\section{Cardinalities estimation under sliding time window}
To calculate the cardinality of all hosts in \textit{ANet}, we only need to extract IP address pairs, like $\{<aip, bip>|aip \in \textit{ANet}, bip \in BNet\}$, from network packets.

When \textit{ANet} is a core network, such as a city-wide network, there are millions of hosts in it. One way to accurately estimate the cardinality of each host is to save and maintain the state of each host and their opposite hosts. For example, an \textit{AT} is assigned to each opposite hosts of $aip$, and the number of active $AT$s ( acquired by $CheckAT()$ operation ) in each time window is counted to obtain the cardinality of $aip$ under the sliding time window. Although this method can get accurate results, it takes up a lot of memory and processing time for the core network. Inspired by GSE, a new algorithm, virtual asynchronous timestamp estimator (\textit{VATE}), is proposed to estimate the cardinality of each host.
\subsection{Virtual asynchronous timestamp estimator}
In cardinality estimation algorithms, each host usually uses the same number of counters. For hosts in \textit{ANet}, their cardinality ranges from 1 to tens of thousands. The larger the cardinality, the more counters the host needs to get an accurate estimation. However, allocating the same number of counters for low-cardinality hosts and high-cardinality hosts causes memory waste. If \textit{AT} can be shared by multiple hosts, the utilization of counters owned by low-cardinality hosts can be improved. \textit{VATE} is put forward under this idea. Let \textit{ATP} denote the \textit{AT} pool consisting of $2^c$ $AT$s and $ATP[i]$ denote the $i$-th \textit{AT} in \textit{ATP}. \textit{VATE} assigns $g$ virtual $AT$s to each host, and each virtual \textit{AT} corresponds to a physical \textit{AT} in the \textit{ATP}. \textit{VATE} is derived from GSE(in Figure \ref{fig_virtural_counter_vector_illustrate}). Unlike GSE, \textit{VATE} uses \textit{AT} instead of bits, as shown in Figure \ref{fig_VATE_AT_pool}.
\begin{figure}[!ht]
\centering
\includegraphics[width=0.47\textwidth]{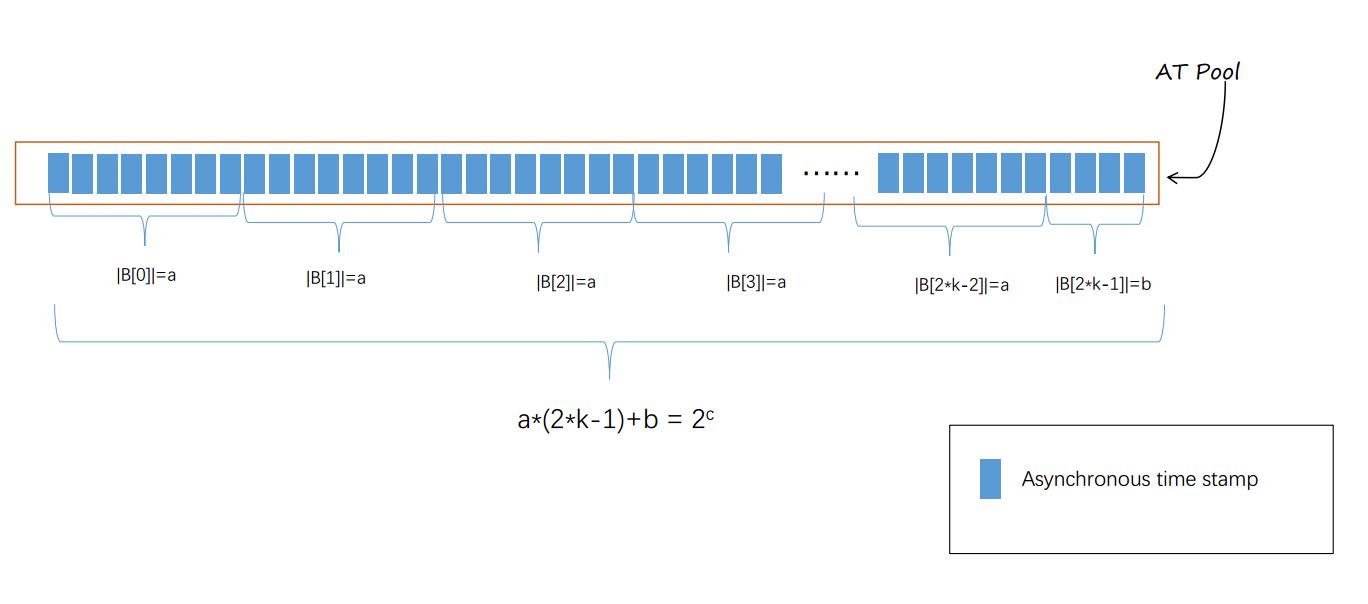}
\caption{Asynchronous time stamp pool}
\label{fig_VATE_AT_pool}
\end{figure}

The state of \textit{AT} in \textit{ATP} is maintained by corollary 1. If each \textit{AT} in \textit{ATP} has the same \textit{ACT}, there are $2^c$ \textit{AT}s needing state-maintenance in these time slices where \textit{ACT} equal to $k$ or 0. But no \textit{AT} needs state-maintenance in other time slices. At this time, the load of the system is unbalanced in different time slices. Moreover, when the value of c is big (for example, $c$ is larger than 27), the time of state-maintenance in these time slices, where the \textit{ACT} is 0 or $k$, will exceed the length of a time slice, and it causes that the algorithm can not run in real time. Therefore, it is necessary to reduce the peak time of state-maintenance operations in all time slices.

Because when to maintain the state of \textit{AT} is determined by its \textit{ACT}, and \textit{ACT} has $2*k$ different values (from 0 to $2*k-1$), \textit{ATP} can be divided according to \textit{ACT} values. In this paper, \textit{AT} in \textit{ATP} is divided into $2*k$ blocks, each \textit{AT} in a block has the same \textit{ACT}. The \textit{ACT} of a \textit{AT} block is defined as the \textit{ACT} of the \textit{AT} in it. Let $B[i]$ denote the $i$-th \textit{AT} block, $|B[i]|$ denote the number of \textit{AT}s in $B[i]$, and $BACT_t^i$ denote the \textit{ACT} of $B[i]$ in time slice $t$, $0\leq i\leq 2*k-1$. The \textit{ACT}s of any two \textit{AT} blocks are different, and the \textit{ACT}s of adjacent \textit{AT} blocks have the following relationship:
\begin{equation}
BACT_t^{i+1}=(BACT_t^i+1) mod (2*k)  
\label{eq_adj_atblock_act}
\end{equation}

In any time slice $t$, when the \textit{ACT} of any \textit{AT} block is known, the \textit{ACT} of other \textit{AT} blocks are calculated according to the formula \ref{eq_adj_atblock_act}. So we only need to record and maintain the \textit{ACT} of an \textit{AT} block. In this paper, we record the \textit{ACT} of the first \textit{AT} block in memory and update its value at the beginning of each time slice. The update method is as follows:
\begin{equation}
 \label{eq_adj_atblock0_act_update}
  BACT_{t+1}^0=(BACT_t^0+1)mod(2*k) 
\end{equation}

In any time slice $t$, $BACT_t^i$ is calculated according to the following formula, where $0 \leq i \leq 2*k-1 $:
\begin{equation}
BACT_t^i=(BACT_t^0+i)mod (2*k)
\label{eq_adj_Get_atblock_act_from_atb0}
\end{equation}

According to the formulas \ref{eq_adj_atblock_act} and \ref{eq_adj_Get_atblock_act_from_atb0}, only two \textit{AT} blocks have \textit{ACT} of 0 or $k$ in any time slice. So only two \textit{AT} blocks in any time slice need the state maintance. To maintain the same number of $AT$s in each time slice, the number of $AT$s contained in each \textit{AT} block should be as same as possible. Because $2^c$ is not necessarily a multiple of $2*k$, when we divide the $2^c$ $AT$s into $2*k$ groups, the number of $AT$s in each group may not be the same. A simple method is to make the number of $AT$s in the previous $2*k-1$ \textit{AT} blocks the same, leaving all the differences in the last \textit{AT} block, as shown in Figure \ref{fig_VATE_AT_pool}. Let $a=|B[i]|$ ($0 \leq i \leq 2*k-2$) and $b=|B[2*k-1]|$, then $a$ and $b$ are calculated according to the following formula.
\begin{equation}
a=\frac{2^c}{2*k-1}
\label{eq-atblk_N_a}
\end{equation}
\begin{equation}
b=2^c\ mod\ (2*k-1) 
\label{eq-atblk_N_b}
\end{equation}

According to Figure \ref{fig_VATE_AT_pool}, the $i$-th \textit{AT} in \textit{ATP} is assigned to $B[i/a]$. This method quickly calculates that an \textit{AT} belongs to which block, but the difference between $|B[2*k-1]|$ and $|B[0]|$ may be more than 1.

Another method is to distribute the differences evenly among the last several \textit{AT} blocks, as shown in Figure \ref{fig_VATE_AT_pool_v2}. Let $a'= 2^c /(2*k)$, $b'= 2^c mod (2*k)$. In the \textit{AT} block partition method shown in Figure \ref{fig_VATE_AT_pool_v2}, for the former $2*k-b'$ \textit{AT} blocks, each contains $a'$ \textit{AT}s; for the latter $b'$ \textit{AT} blocks, each contains $a'+1$ \textit{AT}s. At this time, the difference of the number of $AT$s in any two \textit{AT} blocks will not exceed 1. But when using this partition method, mapping \textit{AT} in \textit{ATP} to \textit{AT} block is more complicated. For the $i$-th \textit{AT} in \textit{ATP}, if $i< a'*(2*k-b'+1)$, $ATP[i]$ belongs to $B[i/a']$; if $i\geq a'*(2*k-b'+1)$, $ATP[i]$ belongs to $B[\frac{i+2*k-b'}{a'+1}]$. 

The \textit{AT} block partition methods shown in figures \ref{fig_VATE_AT_pool} and \ref{fig_VATE_AT_pool_v2} have their own advantages. Considering the simplicity of calculation, the \textit{AT} block partition method shown in Figure \ref{fig_VATE_AT_pool} will be adopted in the following part.
\begin{figure}[!ht]
\centering
\includegraphics[width=0.47\textwidth]{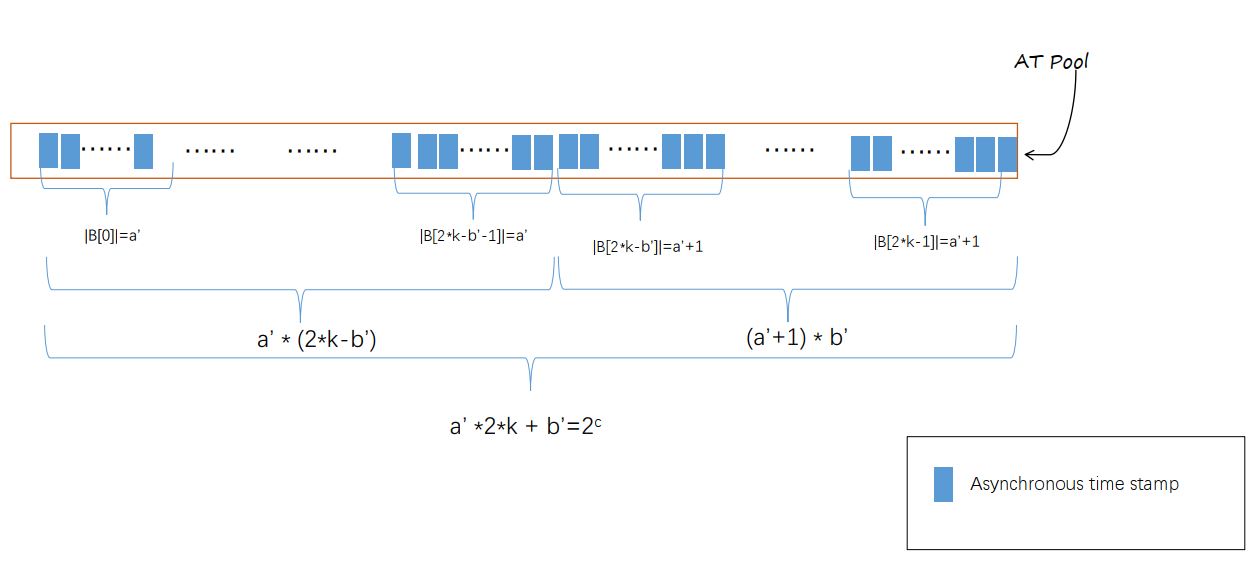}
\caption{Low deviation partition method}
\label{fig_VATE_AT_pool_v2}
\end{figure}

For each $aip$ in \textit{ANet}, \textit{VATE} associates it with $g$ $AT$s in \textit{ATP}. For $aip$, each of its \textit{AT} is called virtual asynchronous timestamp(\textit{VAT}). Let $VAT_{aip}^i$ denote the $i$-th \textit{VAT} of $aip$. $VAT_{aip}^i$ is mapped to a physical \textit{AT} in \textit{ATP} by hash function $H(aip, i)$, that is, $VAT_{aip}^i=ATP[H(aip,i)]$, where $H(aip, i)$ belongs to $[0, 2^c-1$] and $i$ belongs to $[0, g-1]$. The state of \textit{VAT} is the state of its corresponding physical \textit{AT} in \textit{ATP}. For example, if $ATP[H(aip,i)]$ is active, then $VAT_{aip}^i$ is active, otherwise, $VAT_{aip}^i$ is inactive. Each opposite host $bip$ of $aip$ appearing in time slice $t$ will select and set up a \textit{VAT} of $aip$. Let the hash function $BH(bip)$ map $bip$ randomly to an integer between 0 and $g-1$. $BH(bip)$ is used to determine which \textit{VAT} is choosen by $bip$, that is, $ VAT_{aip}^{BH(bip)}$. Setting \textit{VAT} is ultimately the setting of physical \textit{AT} in \textit{ATP}, that is, when $bip$ appears, $SetAT(ATP[H(aip,BH(bip))])$ operation is performed.

According to the attributes of \textit{AT}, at the end of time slice $t$, \textit{VATE} can not only know which $AT$s have been set in time slice $t$, but also determine which $AT$s have been set in the previous $k'-1$ time slices. In another word, at the end of time slice $t$, \textit{VATE} can report which $AT$s are active in time window $W(t-k'+1, k')$, where $k'$ is a positive integer less than or equal to $k$. For \textit{ATP}, the active \textit{AT} can be regarded as the bit ``1" in the algorithm GSE, and the inactive \textit{AT} corresponds to the bit ``0" in the GSE. According to the GSE algorithm, we can get the formula \ref{eq_vate_cardinalityEst} for estimating the cardinality of $aip$ in time window $W(t-k'+1,k')$ at the end of time slice $t$. In time window $W(t-k'+1, k')$, $Z_p^{k'}$ represents the fraction of inactive \textit{AT} in \textit{ATP} and $Z_{aip}^{k'}$ represents the fraction of inactive \textit{VAT} in all \textit{VAT} of $aip$.
\begin{equation}
\label{eq_vate_cardinalityEst}
|OP(aip,t-k'+1,k')|=g*ln(Z_{aip}^{k'})-g*ln(Z_p^{k'})
\end{equation}

Let IPair(t) represent all IP address pairs in time slice $t$. By scanning the IP address pairs in each time slice, \textit{VATE} calculates the host cardinality under the sliding time window. The algorithm \ref{alg-vate_all} fully describes how \textit{VATE} estimates the cardinality of each host in \textit{ANet} under the sliding time window.

\begin{algorithm}                       
\caption{\textit{VATE}}    
\label{alg-vate_all}                 
\begin{algorithmic}[1]
\Require $c$,$k$,$k'$
\State $a \leftarrow \frac{2^c}{k}$ //there are $2^c$ \textit{AT} in \textit{ATP}, and there are most $k$ time slices in a time window.
\For {$at \in ATP$}  \label{alg_vateall_initATP_start}
\State $InitAT(at)$
\EndFor    \label{alg_vateall_initATP_end}
\State $BACT^0 \leftarrow 0$
\State $t \leftarrow 0$ 
\While{(1)}
  \For {$<aip,bip> \in IPair(t)$} \label{alg_vateall_scanIPair_start}
     \State $vid \leftarrow BH(bip)$ \label{alg_vateall_scanIPair_oneThread_start}
     \State $pid \leftarrow H(aip, vid)$
     \State $bid \leftarrow \frac{pid}{a}$
     \State $ACT^{ATP[pid]}_t \leftarrow (BACT^0 + bid) mod (2*k)$
     \State $SetAT(ATP[pid])$  \label{alg_vateall_scanIPair_oneThread_end}
   \EndFor  \label{alg_vateall_scanIPair_end}
   \For {$aip \in ANet$}  \label{alg_vateall_estCardinality_start}   
   \State $Cardinality \leftarrow g*ln(Z^{k'}_v)-g*ln(Z^{k'}_p)$\label{alg_vateall_estCardinality_oneThread_start} 
   \State Output $aip$ and its $Cardinality$\label{alg_vateall_estCardinality_oneThread_end} 
   \EndFor  \label{alg_vateall_estCardinality_end}
   \State $BACT^0 \leftarrow (BACT^0 +1 ) mod (2*k)$ \label{alg_vateall_updateBACT0}
   \State $t \leftarrow t+1$ \label{alg_vateall_timeslice_sliding}
   
   \For { $ i \in [0, 2*k-1]$}\label{alg_vateall_preserveATB_start}
     \State $BACT_t^i \leftarrow (BACT^0 + i) mod (2*k)$ 
     \If {$(BACT_t^i == 0)$ or $(BACT_t^i == k)$}
       \For {$at \in B[i]$}
       \State $SetAT(at)$
       \EndFor
     \EndIf  
   \EndFor\label{alg_vateall_preserveATB_end}
\EndWhile
\end{algorithmic}
\end{algorithm}

\textit{VATE} initialises all $AT$s in \textit{ATP} first. Then start scanning IP address pairs in each time slice. For each IP address pair in time slice $t$, lines \ref{alg_vateall_initATP_start} to \ref{alg_vateall_initATP_end} of the algorithm \ref{alg-vate_all} map it to an \textit{AT} in \textit{ATP} and set the \textit{AT}. After scanning all IP address pairs in IPair(t), \textit{VATE} begins to calculate and output the cardinality of each $aip$ in \textit{ANet} in the time window $W(t-k'+1, k')$ (from line \ref{alg_vateall_estCardinality_start} to \ref{alg_vateall_estCardinality_end}). After calculating the cardinality of all $aip$ in \textit{ANet}, \textit{VATE} updates the \textit{ACT} of the first \textit{AT} block in line \ref{alg_vateall_updateBACT0}, and then slides the window by a time slice. After the window sliding, \textit{VATE} maintains all $AT$s in the \textit{AT} block whose \textit{BACT} is 0 or $k$, as shown from line \ref {alg_vateall_preserveATB_start} to line \ref {alg_vateall_preserveATB_end} in algorithm \ref{alg-vate_all}.

\textit{VATE} can run continuously with only one initialisation. And \textit{VATE} is suitable for running on a parallel computing platform to process high-speed network data in real time. The next section describes how to deploy \textit{VATE} on the GPU.

\subsection{Deploy \textit{VATE} on GPU}
In high-speed networks, such as 40 Gb/s, millions of packets pass through the edges of the network every second. Scanning so many packets in real time requires a lot of computing resources. CPU is one of the most common computing components. Each core of CPU can handle complex tasks running different instructions. Although the computing core of CPU is powerful, its price is very high. If we want to use hundreds of CPU computing cores to process high-speed network data, we must adopt a cluster composed of multiple CPUs. The cost of the cluster will increase with the increase of scale. Graphics Processing Unit (GPU) is one of the most popular parallel computing platforms in recent years. A GPU chip contains hundreds to thousands of processing units, far more than that in the CPU. For tasks without data access conflicts and using the same instructions to process different data (single instruction multiple data streams, SIMD), GPU can achieve high speedup\cite{PD2013:BenchmarkingOfCommunicationTechniquesForGPUs}\cite{PD2013:GeneratingDataTransfersForDistributedGPUParallelPrograms}.

As can be seen from the algorithm \ref{alg-vate_all}, \textit{VATE} has three main steps: scanning IP address pairs, estimating cardinality and maintaining state. Each step uses the same steps to process different data (IP address pairs, host in \textit{ANet}, \textit{AT} in \textit{ATP}), and there is no read-write conflict between different data processing. So \textit{VATE} is a SIMD program. When \textit{VATE} runs on GPU, $s$ threads can be used to process $s$ data at the same time. The value of $s$ is usually more than one thousand. Figure \ref{fig_VATE_RunOnGPU} illustrates how \textit{VATE} runs on GPU.

\begin{figure}[!ht]
\centering
\includegraphics[width=0.47\textwidth]{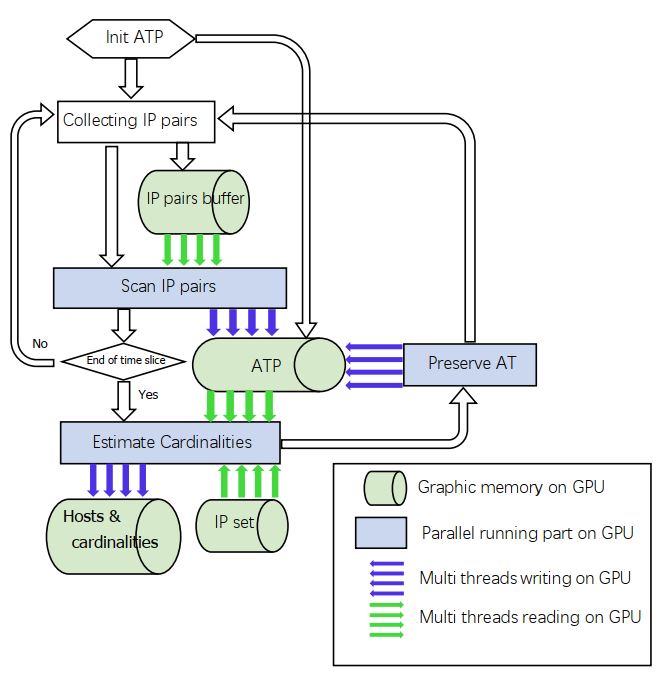}
\caption{\textit{VATE} running on GPU}
\label{fig_VATE_RunOnGPU}
\end{figure}

\textit{VATE} first allocates $2^c*ceil(log_2(2*k+1))$ bits on GPU graphical memory to store \textit{ATP} and initialise each \textit{AT} in the \textit{ATP}. Then \textit{VATE} starts processing IP address pairs in each time slice. GPU is connected to the server through PIC bus. Since threads in the GPU can only access the graphics memory of the GPU, IP address pairs must be copied from the server to the graphics memory of the GPU. Network packets come in turn. It is not efficient to extract and copy IP address pair one by one, because each data transmission session between server memory and GPU memory requires additional starting and ending operations. To improve efficiency, we allocate two buffers to store $r$ IP address pairs on the server side and GPU side respectively. For IPv4 addresses, an IP address pair takes 8 bytes, so the size of each IP address pair buffer is $r*8$ bytes. On the server, after extracting $r$ IP address pairs, \textit{VATE} copies them to the buffer on the GPU. At the end of a time slice, if the IP address pairs buffer on the server is not full, it will be filled with the last IP address pair in the buffer, and then copied to the GPU. The larger the $r$, the more efficient the copying process is. However, the delay of scanning IP address pairs will also increase, because the scanning process of IP address pairs will not run until the IP address pairs on the GPU receive $r$ IP address pairs on the buffer. And under the sliding time window, IP address pairs are divided into time slices. If $r$ is greater than the number of IP address pairs in a time slice, the extra part will cause memory waste. The value of $r$ in this article is set to $2^{15}$.

When the IP address pair buffer on the GPU receives $r$ IP address pairs, \textit{VATE} enters the ``Scanning IP address pairs" process. This process starts $r$ GPU threads, each GPU thread reads an IP address pair in the buffer, and then processes the IP address pair according to the steps from line \ref{alg_vateall_scanIPair_oneThread_start} to \ref{alg_vateall_scanIPair_oneThread_end} of algorithm \ref{alg-vate_all}.

At the end of a time slice, the ``estimating cardinality" process begins. In this process, \textit{VATE} first divides the IP addresses in \textit{ANet} into groups and each group contains $s$ IP addresses. The number of IP addresses in the last group may be less than $s$, and the insufficient part can be filled with 0. Then \textit{VATE} calculates the cardinality of IP addresses in \textit{ANet} group by group on the GPU. For each IP group, \textit{VATE} starts $s$ threads on the GPU. Each thread calculates the cardinality of an IP address by the steps from line \ref{alg_vateall_estCardinality_oneThread_start} to \ref{alg_vateall_estCardinality_oneThread_end} of algorithm \ref{alg-vate_all}.

After finishing “estimating cardinality” process, \textit{VATE} slides the time window, updates $BACT^0$, and then begins to maintain the state of \textit{AT}s in two blocks. \textit{VATE} divides \textit{AT} in these two \textit{AT} blocks into groups, and each group contains $s$ \textit{AT}(the number of \textit{AT} in the last group may smaller than $s$, and the deficiencies will not be dealt with). Then \textit{VATE} maintains the state of \textit{AT} group by group on GPU. For every group, \textit{VATE} launches $s$ threads, and each thread preserves the state of an \textit{AT}(apply $PreserveAT()$ operation).

After state-maintenance, \textit{VATE} starts to scan IP address pairs in the next time slice. In this way, \textit{VATE} continuously estimates the cardinality of each IP address under the sliding time window on GPU in real time.

As can be seen from the Figure \ref{fig_VATE_RunOnGPU}, the three main processes of \textit{VATE} are executed in sequence. Each process handles multiple data concurrently on the GPU. And each process does not read and write the same memory area at the same time. So \textit{VATE} does not have read-write conflicts when running in parallel. Only in the ``scanning IP address pairs" process, multiple threads may write an \textit{AT} ($SetAT()$, set the \textit{AT} value to 0) at the same time. However, various threads set the same \textit{AT}, regardless of the order in which these threads write the \textit{AT}, the result is the same (\textit{AT} value is set to 0). So \textit{VATE} does not have write conflicts when running in parallel. 

When running on GPU, even if \textit{ATP} contains up to $2^{32}$ $AT$s, \textit{VATE} can maintain the state of each \textit{AT} in real time and estimate the cardinality of each host under the sliding time window. In the next section, we will demonstrate the advantages of the \textit{VATE} algorithm in memory usage and state-maintenance through experiments on real high-speed network traffic.

\section{Experiments}
In order to evaluate the performance of \textit{VATE}, we downloaded traffic data for 1 hour from a 10 Gb/s core network (Caida\cite{expdata:Caida}) and a 40 Gb/s core network (IPtas\cite{expdata:IPtraceCernetJS}) respectively. Caida's traffic was the network packets between Seattle and Chicago from 13:00 to 14:00 on February 19, 2015. IPtas traffic was the network packets between CERNET Jiangsu node and other networks from 13:00 to 14:00 on October 23, 2017. In these experiments, for Caida traffic, the hosts in Seattle composes \textit{ANet}, the hosts in Chicago composes \textit{BNet}; for IPtas traffic, hosts in CERNET Jiangsu node make up \textit{ANet}, hosts in other networks make up \textit{BNet}.

Table \ref{fig_exp_traffic_static} lists the main characteristics of these two sets of traffics in a five-minute time window. In table \ref{fig_exp_traffic_static}, ``$\#AIP$" denotes the number of hosts in \textit{ANet}, ``$\#BIP$" denotes the number of hosts in \textit{BNet}, and ``$\#Flow$" denotes the number of different IP address pairs. ``Speed(kpps)" denotes the speed of packets with unit one thousand packets per second. ``$\#IPc100$" denotes the number of hosts whose cardinality is greater than 100, and ``$\%IPc100$" denotes the percentage of hosts whose cardinality is greater than 100.
\begin{figure}[!ht]
\centering
\includegraphics[width=0.47\textwidth]{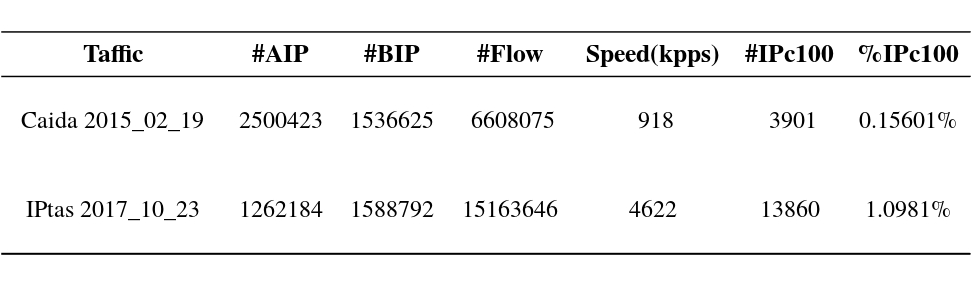}
\caption{Traffic summary}
\label{fig_exp_traffic_static}
\end{figure}

As can be seen from table \ref{fig_exp_traffic_static}, hosts with cardinalities greater than 100 account for only a small part of all hosts. In practical applications, researchers and managers are interested in those hosts with high cardinalities. Moreover, cardinality estimation for a large number of small-cardinality hosts will take too much computing time and reduce the efficiency of an algorithm. So this experiment only focuses on those hosts whose cardinalities are greater than or equal to 100.

In our experiment, we set the size of a time slice to 1 second, and the first time slice lasting from 13:00:00 to 13:00:01 is regarded as time slice 0. $k'$ is set to 300, that is, to calculate the cardinality in 5 minutes. In a five-minute window, the cardinality of most hosts is less than 5000. For example, Figure \ref{fig_exp_CardinalityDistribution_s600_v2} shows the distribution of cardinalities of these two traffic in $W(600,300)$.
\begin{figure}[!ht]
\centering
\includegraphics[width=0.47\textwidth]{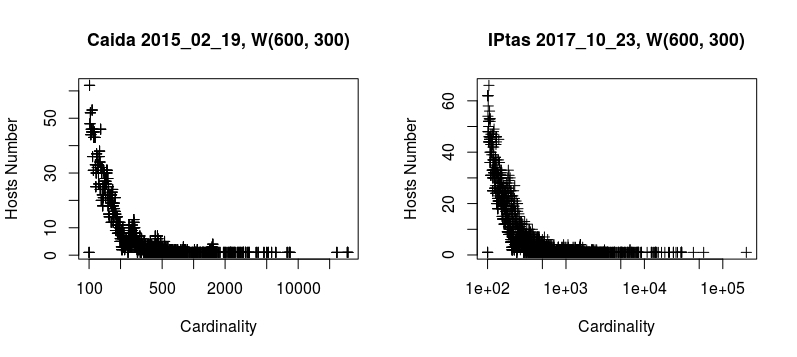}
\caption{Cardinality distribution in a time window}
\label{fig_exp_CardinalityDistribution_s600_v2}
\end{figure}

In the time window $W(600,300)$, there are only 10 and 72 hosts respectively whose cardinalities are greater than or equal to 5000 in $Caida\ 2015\_02\_19$ and $IPtas\ 2017\_10\_23$. Cardinalities of most hosts($99.733\%$ in $Caida\ 2015\_02\_19$ and $98.816\%$ in $IPtas\ 2017\_10\_23$) are less than 5000.\textit{VATE} is based on GSE algorithm. According to GSE, the larger the number of virtual counters $g$, the larger cardinality that can be estimated by \textit{VATE}. But the larger the $g$, the more time it takes to estimate the cardinality because \textit{VATE} needs to check the states of more virtual counters. To compare the effects of the number of physical counters (for \textit{VATE}, the number of counters is the number of $AT$s in \textit{ATP}, i.e. $2^c$) and the number of virtual counters($g$) on cardinality estimation, we tested estimation accuracy of \textit{VATE} under different parameters. Figure \ref{fig_exp_ATV_est_real_diffCN_w600_data1_v2} and \ref{fig_exp_ATV_est_real_diffCN_w600_data10_v2} show the cardinalities of $Caida\ 2015\_02\_19$ and $IPtas\ 2017\_10\_23$ estimated by \textit{VATE} under time window W(600,300). In these two graphs, the abscissa represents the actual cardinality, and the ordinate is the cardinality estimated by \textit{VATE}. When the points in the figure are closer to the oblique line(the oblique line composed of the points whose estimated cardinality equals the actual cardinality), the accuracy of the estimated result is higher. When the point in the graph is below the right of the oblique line, it means that the cardinality is underestimated; when the point is above the left of the oblique line, it means that the cardinality is overestimated.
\begin{figure*}[!ht]
\centering
\includegraphics[width=0.97\textwidth]{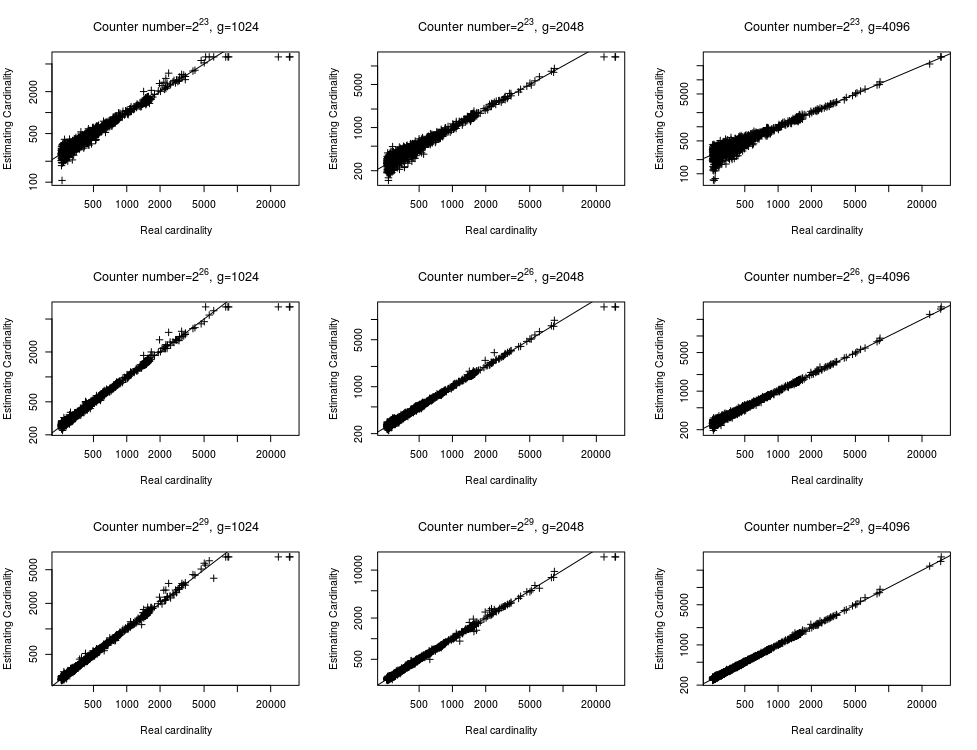}
\caption{Cardinality estimation of Caida $2015\_02\_19$}
\label{fig_exp_ATV_est_real_diffCN_w600_data1_v2}
\end{figure*}

\begin{figure*}[!ht]
\centering
\includegraphics[width=0.97\textwidth]{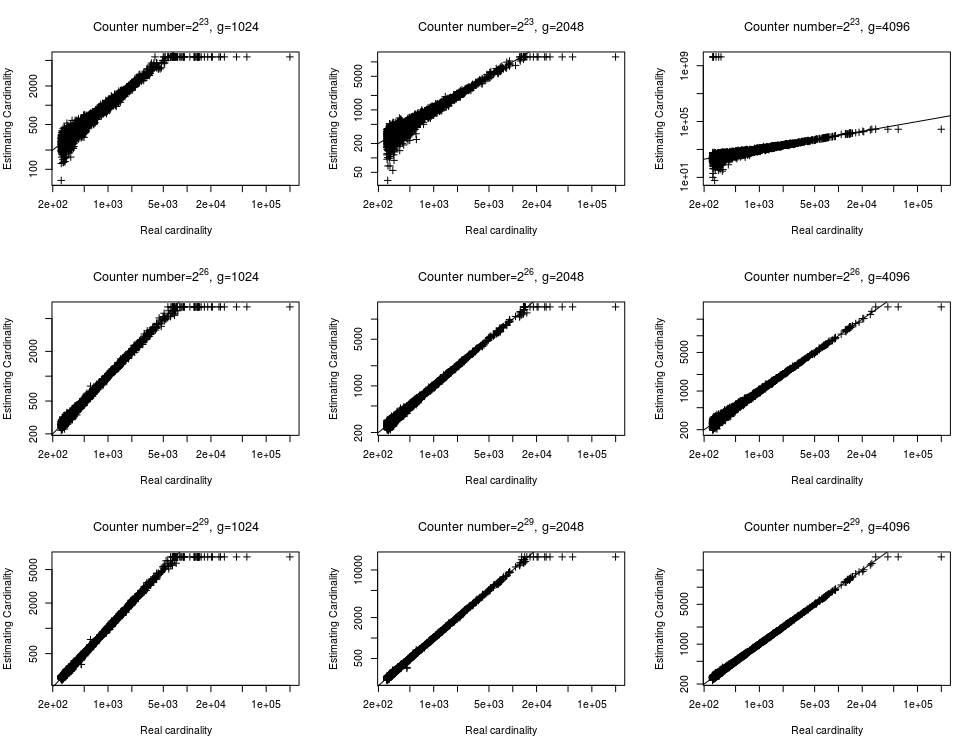}
\caption{Cardinality estimation of IPtas $2017\_10\_23$}
\label{fig_exp_ATV_est_real_diffCN_w600_data10_v2}
\end{figure*}

As can be seen from Figure \ref{fig_exp_ATV_est_real_diffCN_w600_data1_v2} and \ref{fig_exp_ATV_est_real_diffCN_w600_data10_v2}, when the number of counters remains unchanged ($2^{23}$, $2^{26}$ or $2^{29}$) and $g$ increases from 1024 to 4096, the number of hosts whose cardinality cannot be approximated by \textit{VATE} is also decreases. When $g = 4096$, for $Caida\ 2015\_02\_19$, the cardinality of all hosts in the time window W(600,300) can be approximately estimated. When the value of $g$ is constant, the larger the number of counters, the closer the points are to the oblique line, that is, the closer the estimated value is to the actual cardinality. Therefore, to improve the efficiency of \textit{VATE}, $g$ should be determined according to these large cardinalities in the window; in order to improve the accuracy of the algorithm, more counters should be stored in the counter pool. But the more counters there are, the more memory the counter pool occupies and the more time it takes to maintain the state of the counter.

\textit{DR} is a counter under the sliding time window with optimal memory and small state-maintenance operation. Similar to \textit{VATE}, when the counter in the counter pool is \textit{DR}, the \textit{VDRE} algorithm can be obtained (the setting of \textit{AT}, state-maintenance and state checking become the corresponding operation of \textit{DR}). Both \textit{VATE} and \textit{VDRE} are based on GSE. When the number of counters in the counter pool is the same, and the mapping method of IP address to virtual counters and virtual counters to physical counters in the counter pool are the same, \textit{VATE} and \textit{VDRE} will get the same estimating value. The difference lies in memory usage and state-maintaining time. When $k'$ is set to 300, each \textit{DR} takes 9 bits ($ceil(log_2(k'+1)$) and each \textit{AT} takes 10 bits ($ceil(log_2(2*k'+1)$). Although \textit{AT} contains one bit more than \textit{DR}, \textit{AT} has less state-maintaining time. The running time of \textit{VATE} and \textit{VDRE} can be divided into three parts: packet scanning time (ST), cardinality estimation time (ET) and counter state-maintaining time (PT). Different running platforms influence these three kinds of running time. We use three different Nvidia GPUs to test the running time of these algorithms: GTX650 with 1 GB graphic memory, GTX950 with 4 GB graphic memory and TitanXp with 12 GB graphic memory. The computing power and graphics memory of these three GPUs are improved in sequence. The three GPUs are represented by GTX650-1GB, GTX950-4GB and TitanXP-12GB, respectively. Tables \ref{fig_exp_allConsumingTime_onDifferentGPU_100131_v1} and \ref{fig_exp_allConsumingTime_onDifferentGPU_100132_v1} show the average running time(avgST, avgET and avgPT) of \textit{VDRE} and \textit{VATE} when estimating the cardinalities of these hosts in $Caida\ 2015\_02\_19$ on the three GPU platforms. Among them, $g$ is set to 1024, and the number of counters increases from $2^{23}$ to $2^{32}$. “$\#Cnt.(lg2)$” is the exponent of the number of counters when the base number is 2.
\begin{table}
\centering
\caption{Time consumption of \textit{VDRE} on different GPU}
\label{fig_exp_allConsumingTime_onDifferentGPU_100131_v1}
\begin{tabular}{c}                                                                                                                                                                                                                           
\centering
\includegraphics[width=0.45\textwidth]{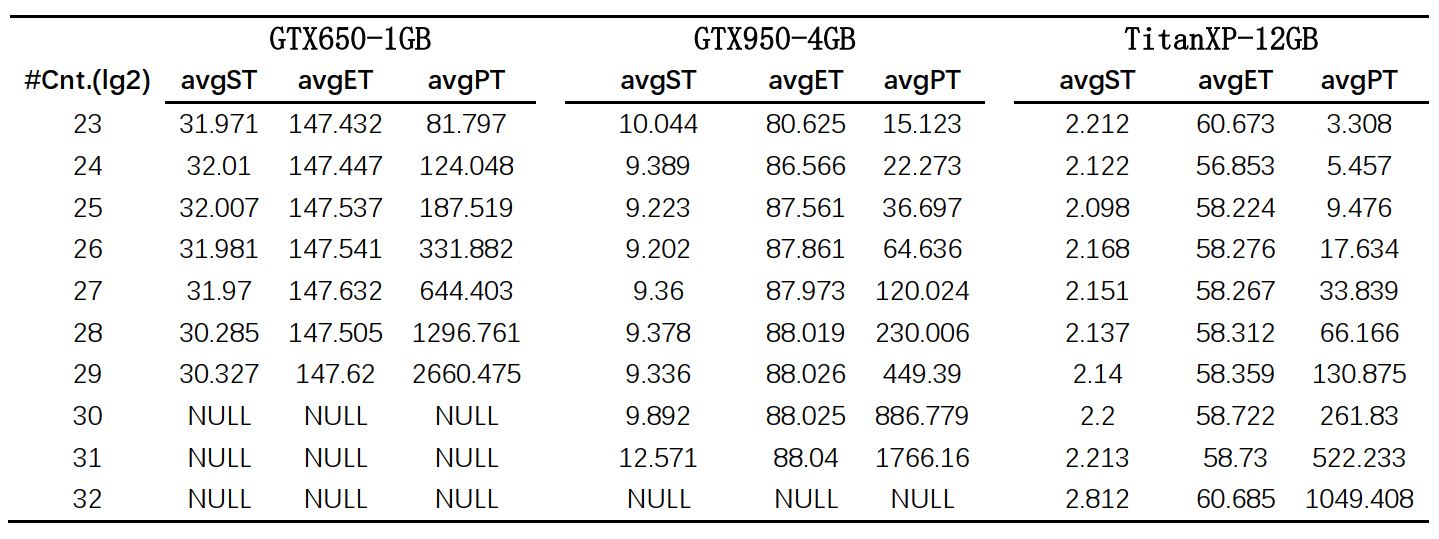}
\end{tabular}
\end{table}

\begin{table}
\centering
\caption{Time consumption of \textit{VATE} on different GPU}
\label{fig_exp_allConsumingTime_onDifferentGPU_100132_v1}
\begin{tabular}{c}                                                                                                                                                                                                                           
\centering
\includegraphics[width=0.45\textwidth]{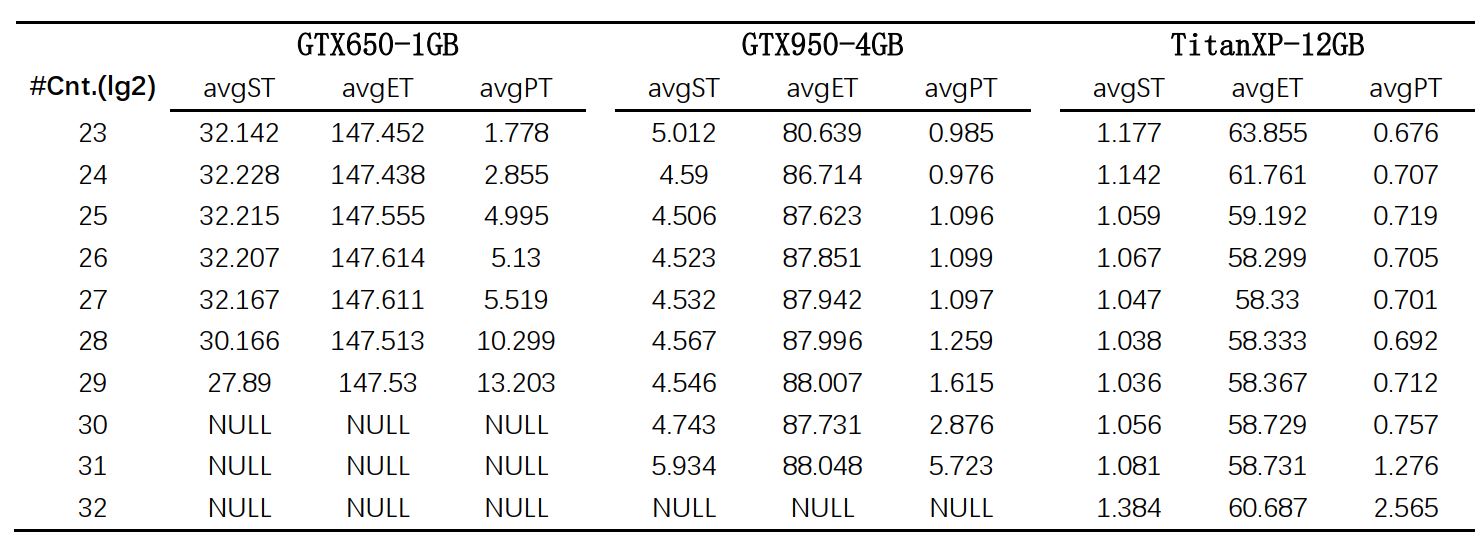}
\end{tabular}
\end{table}

In tables \ref{fig_exp_allConsumingTime_onDifferentGPU_100131_v1} and \ref{fig_exp_allConsumingTime_onDifferentGPU_100132_v1}, NULL indicates that the counter pool is larger than the GPU memory and cannot run on the GPU. As can be seen from these two tables, the three running times of \textit{VDRE} and \textit{VATE} are greatly reduced with the improvement of GPU computing power. According to algorithm \ref{alg-vate_all}, ST and ET are not affected by the number of counters. So on the same GPU, no matter how many counters are contained in the counter pool, there is little change of ST or ET. The difference between \textit{VDRE} and \textit{VATE} is reflected in PT. PT increases with the number of counters in the counter pool, as shown in Figure \ref{fig_exp_preserveTime_cmp_list_diffCN_diffGPU_linePoints}. For \textit{VDRE}, this upward trend is more obvious.
\begin{figure*}[!ht]
\centering
\includegraphics[width=0.97\textwidth]{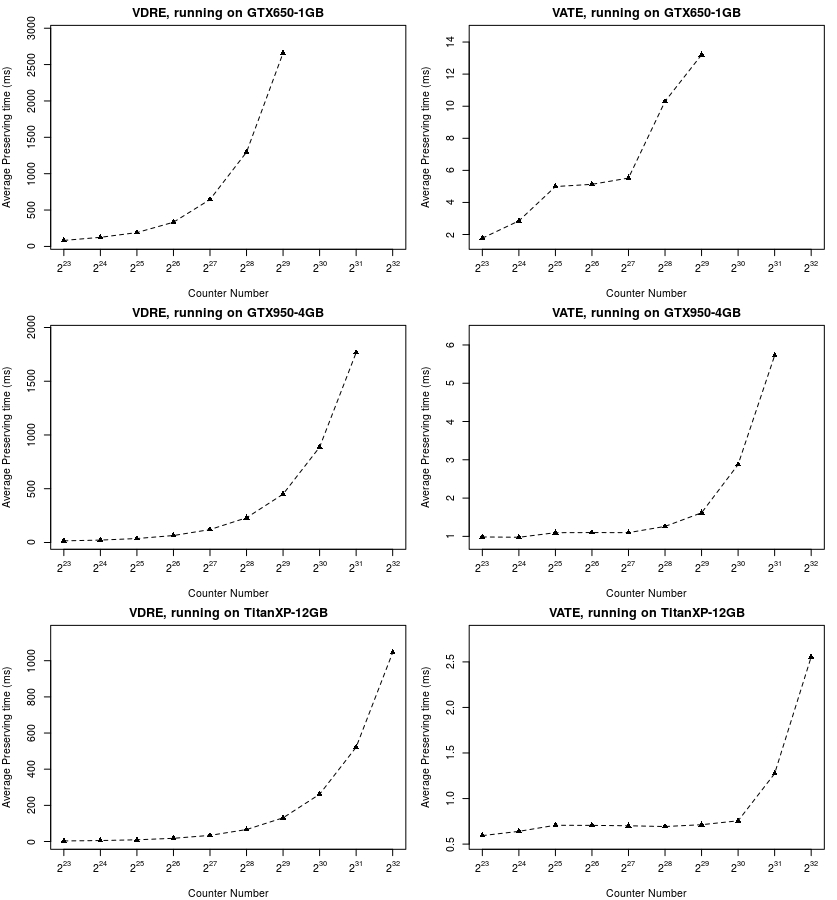}
\caption{Preserving time comparing on different GPU}
\label{fig_exp_preserveTime_cmp_list_diffCN_diffGPU_linePoints}
\end{figure*}

On the same GPU, when the counter pool contains the same number of counters, \textit{VATE}'s PT is significantly less than \textit{VDRE}'s (\textit{VATE}'s PT accounts for only $25\%$ to $0.25\%$ of \textit{VDRE}'s PT). On GTX650-1GB, when the number of counters is $2^{28}$, the PT of \textit{VDRE} is as high as 1296 milliseconds, and the sum of three running times is 1447 milliseconds. For the algorithm under the sliding time window, the total running time in each time slice should not exceed the length of a time slice. Otherwise the algorithm can not run in real time. In this experiment, the length of the time slice is 1 second, while when the number of counters is $2^{28}$ on GTX650-1GB, the running time of processing data in each time slice is 1.4 times of the length of time slice. TitanXP-12GB has more computing power than GTX650-1GB, which can effectively reduce the running time. But the price of TitanXP-12GB is also higher than that of GTX650-1GB, and when the number of counters reaches $2^{32}$, even on TitanXP-12GB, \textit{VDRE} cannot run in real time, because its PT exceeds one second. So it is not the GPU memory that limits the accuracy of \textit{VDRE}, but the state-maintaining time of \textit{VDRE}. \textit{VATE} solves this problem very well. As can be seen from the table, the PT of \textit{VATE} does not exceed 14 milliseconds on any GPU.

On different GPUs, when the number of counters takes different values, the difference between \textit{VATE}'s PT and \textit{VDRE}'s PT is also changed. To measure this difference, we define the PT speedup as the ratio of  \textit{VDRE}'s PT to \textit{VATE}'s PT on the same GPU with the same number of counters. Table \ref{fig_exp_preserveTime_speed_up_v1} shows the change of PT speedup with different counter number on different platforms. As can be seen from the table, \textit{VATE} can achieve better PT speedup for low-computing capacity GPUs. On the same GPU, the greater the number of counters, the greater the PT speedup. On TitanXP-12GB with $2^{32}$ counters, \textit{VATE} even acquires a PT speedup as high as 409.
\begin{table}
\centering
\caption{Preserving time speed up of \textit{VATE}}
\label{fig_exp_preserveTime_speed_up_v1}
\begin{tabular}{c}                                                                                                                                                                                                                           
\centering
\includegraphics[width=0.45\textwidth]{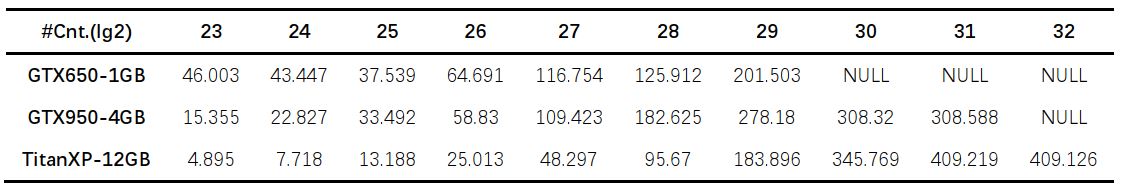}
\end{tabular}
\end{table}

When the number of counters is constant, the PT in each time slice is not necessarily the same on the same GPU. PT of each time slice will fluctuate within a certain range. Figure \ref{fig_exp_preserveTime_cmp_oneCounterN_AllGPU_v1} shows the PT of each time slice on different GPUs when the number of counters is $2^{29}$. As can be seen from Figure \ref{fig_exp_preserveTime_cmp_oneCounterN_AllGPU_v1}, PT changes greatly at the beginning, and then gradually stabilizes. The higher the computing power of GPU, the lower the PT variance of different time slices. On the same GPU, the variance of \textit{VATE}'s PT is lower than \textit{VDRE}'s PT. Therefore, the fluctuation of \textit{VATE}'s PT is small, and it runs smoothly. It can also be seen from the figure that in each time slice, the PT of \textit{VATE} is smaller than that of \textit{VDRE}. Regardless of the computing capacity of GPU, \textit{VATE} can make full use of the graphics memory of GPU to estimate cardinalities in real time with higher accuracy.
\begin{figure*}[!ht]
\centering
\includegraphics[width=0.97\textwidth]{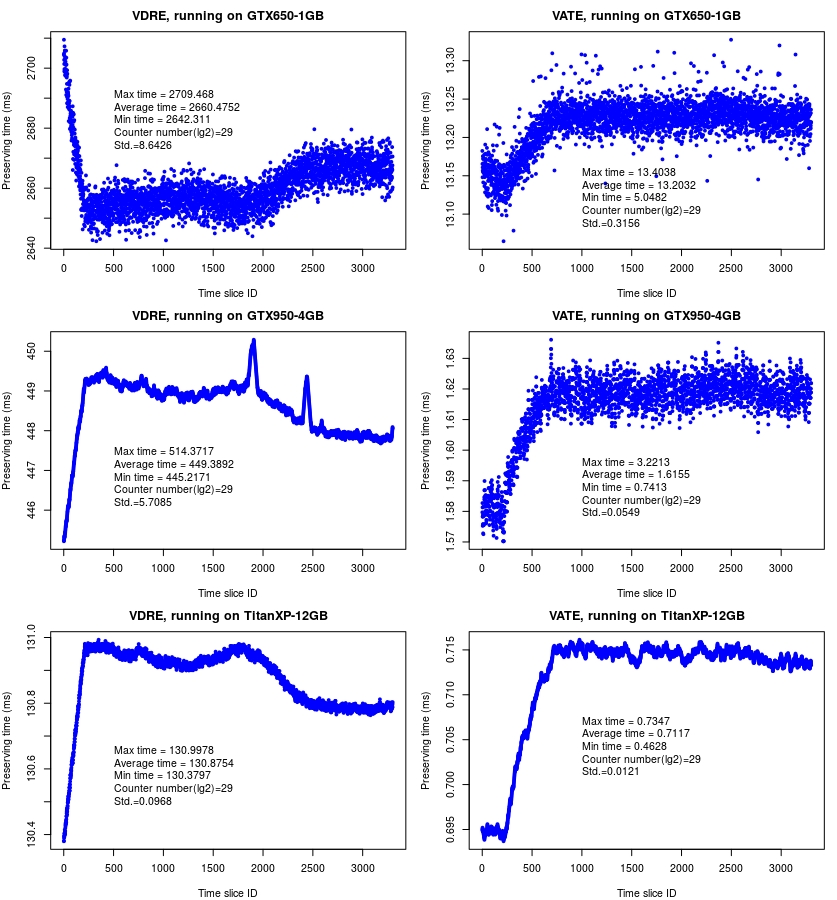}
\caption{Preserving time in different time slices}
\label{fig_exp_preserveTime_cmp_oneCounterN_AllGPU_v1}
\end{figure*}

\section{Conclusion}
The \textit{AT} proposed in this paper has the advantages of less memory consumption and low state-maintaining time. For sliding time window with $k$ time slices at most, each \textit{AT} only takes up $ceil(log_2(2*k+1))$ bits, and the time complexity of maintaining \textit{AT} state is only $O(1/k)$. Based on \textit{AT}, this paper designs a new multi-hosts cardinalities estimation algorithm \textit{VATE}. \textit{VATE} allocates a virtual \textit{AT} vector for each host to estimate the cardinality. \textit{VATE} improves memory utilisation by sharing $AT$s in an \textit{AT} pool. The higher the number of $AT$s in \textit{ATP}, the higher the accuracy of \textit{VATE}. Because \textit{AT} occupies less memory and requires smaller state-maintaining time, \textit{VATE} can use more counters than existing algorithms to estimate cardinalities of different hosts under sliding time windows in real time with higher accuracy. \textit{VATE} is also a parallel algorithm that can be deployed on GPU. With the parallel processing capacity of GPU, \textit{VATE} can estimate the cardinalities of hosts in a 40 Gb/s high-speed network with the length of time slice as small as 1 second. In the further work, we will deploy \textit{VATE} in CERNET network to provide critical host attribute for network security and management. 
 \section*{Acknowledgment}
The authors would like to thank anonymous reviewers. The research work leading to this article is supported by the National Natural Science Foundation of China under Grant No. 61602114

\iftoggle{ACM}{
\bibliographystyle{ACM-Reference-Format}
}
\iftoggle{IEEEcls}{
\bibliographystyle{IEEEtran}
}
\iftoggle{ElsJ}{
\bibliographystyle{elsarticle-num}
}

\bibliography{..//ref} 

\begin{thebibliography}{10}
\expandafter\ifx\csname url\endcsname\relax
  \def\url#1{\texttt{#1}}\fi
\expandafter\ifx\csname urlprefix\endcsname\relax\def\urlprefix{URL }\fi
\expandafter\ifx\csname href\endcsname\relax
  \def\href#1#2{#2} \def\path#1{#1}\fi

\bibitem{DosC:ACooperativeIntrusionDetectionFrameworkCloud}
C.~C. Lo, C.~C. Huang, J.~Ku, A cooperative intrusion detection system
  framework for cloud computing networks, in: 2010 39th International
  Conference on Parallel Processing Workshops, 2010, pp. 280--284.
\newblock \href {http://dx.doi.org/10.1109/ICPPW.2010.46}
  {\path{doi:10.1109/ICPPW.2010.46}}.

\bibitem{DosC:AnalysisSimulationDDOSAttackCloud}
S.~Karthik, J.~J. Shah, Analysis of simulation of ddos attack in cloud, in:
  International Conference on Information Communication and Embedded Systems
  (ICICES2014), 2014, pp. 1--5.
\newblock \href {http://dx.doi.org/10.1109/ICICES.2014.7033841}
  {\path{doi:10.1109/ICICES.2014.7033841}}.

\bibitem{Scan:EvasionResistantNetworkScanDetection}
R.~E. Harang, P.~Mell,
  \href{http://dx.doi.org/10.1186/s13388-015-0019-7}{Evasion-resistant network
  scan detection}, Security Informatics 4~(1) (2015) 4.
\newblock \href {http://dx.doi.org/10.1186/s13388-015-0019-7}
  {\path{doi:10.1186/s13388-015-0019-7}}.
\newline\urlprefix\url{http://dx.doi.org/10.1186/s13388-015-0019-7}

\bibitem{HSD:sampleFlowDistributionEstimate}
N.~Duffield, C.~Lund, M.~Thorup,
  \href{http://doi.acm.org/10.1145/863955.863992}{Estimating flow distributions
  from sampled flow statistics}, in: Proceedings of the 2003 Conference on
  Applications, Technologies, Architectures, and Protocols for Computer
  Communications, SIGCOMM '03, ACM, New York, NY, USA, 2003, pp. 325--336.
\newblock \href {http://dx.doi.org/10.1145/863955.863992}
  {\path{doi:10.1145/863955.863992}}.
\newline\urlprefix\url{http://doi.acm.org/10.1145/863955.863992}

\bibitem{HSD:streamingAlgorithmFastDetectionSuperspreaders}
S.~Venkataraman, D.~Song, P.~B. Gibbons, A.~Blum, New streaming algorithms for
  fast detection of superspreaders, in: in Proceedings of Network and
  Distributed System Security Symposium (NDSS, 2005, pp. 149--166.

\bibitem{HSD:AcontinuousVirtualVectorBasedAlgorithmMeasuringCardinalityDistribution}
X.~Zhou, W.~Liu, Z.~Li, W.~Gao, A continuous virtual vector-based algorithm for
  measuring cardinality distribution, in: Algorithms and Architectures for
  Parallel Processing: 14th International Conference, ICA3PP 2014, Dalian,
  China, August 24-27, 2014. Proceedings, Part II, Springer International
  Publishing, Cham, 2014, pp. 43--53.

\bibitem{SDC:IMC2003:IdentifyingFrequentItemsSlidingWindowsOnlinePacketStreams}
L.~Golab, D.~DeHaan, E.~D. Demaine, A.~Lopez-Ortiz, J.~I. Munro,
  \href{http://doi.acm.org/10.1145/948205.948227}{Identifying frequent items in
  sliding windows over on-line packet streams}, in: Proceedings of the 3rd ACM
  SIGCOMM Conference on Internet Measurement, IMC '03, ACM, New York, NY, USA,
  2003, pp. 173--178.
\newblock \href {http://dx.doi.org/10.1145/948205.948227}
  {\path{doi:10.1145/948205.948227}}.
\newline\urlprefix\url{http://doi.acm.org/10.1145/948205.948227}

\bibitem{SDC:MaintainingStreamStatisticsOverSlidingWindows}
M.~Datar, A.~Gionis, P.~Indyk, R.~Motwani,
  \href{https://doi.org/10.1137/S0097539701398363}{Maintaining stream
  statistics over sliding windows}, SIAM J. Comput. 31~(6) (2002) 1794--1813.
\newblock \href {http://dx.doi.org/10.1137/S0097539701398363}
  {\path{doi:10.1137/S0097539701398363}}.
\newline\urlprefix\url{https://doi.org/10.1137/S0097539701398363}

\bibitem{ISPA2017:HighSpeedNetworkSuperPointsDetectionBasedSlidingWindowGPU}
J.~Xu, W.~Ding, J.~Gong, X.~Hu, J.~Liu, High speed network super points
  detection based on sliding time window by gpu, in: 2017 IEEE International
  Symposium on Parallel and Distributed Processing with Applications and 2017
  IEEE International Conference on Ubiquitous Computing and Communications
  (ISPA/IUCC), 2017, pp. 566--573.
\newblock \href {http://dx.doi.org/10.1109/ISPA/IUCC.2017.00092}
  {\path{doi:10.1109/ISPA/IUCC.2017.00092}}.

\bibitem{PCSA:ProbabilisticCountingAlgorithmsForDataBaseApplications}
P.~Flajolet, G.~N. Martin,
  \href{http://www.sciencedirect.com/science/article/pii/0022000085900418}{Probabilistic
  counting algorithms for data base applications}, Journal of Computer and
  System Sciences 31~(2) (1985) 182 -- 209.
\newblock \href
  {http://dx.doi.org/https://doi.org/10.1016/0022-0000(85)90041-8}
  {\path{doi:https://doi.org/10.1016/0022-0000(85)90041-8}}.
\newline\urlprefix\url{http://www.sciencedirect.com/science/article/pii/0022000085900418}

\bibitem{DC:LoglogCountingOfLargeCardinalitiesDurand2003}
M.~Durand, P.~Flajolet, Loglog counting of large cardinalities, in:
  G.~Di~Battista, U.~Zwick (Eds.), Algorithms - ESA 2003: 11th Annual European
  Symposium, Budapest, Hungary, September 16-19, 2003. Proceedings, Springer
  Berlin Heidelberg, Berlin, Heidelberg, 2003, pp. 605--617.

\bibitem{DC:HyperLogLogTheAnalysisOfANearoptimalCardinalityEstimationAlgorithm}
P.~Flajolet, E.~Fusy, O.~Gandouet, F.~Meunier,
  \href{https://hal.archives-ouvertes.fr/hal-00406166}{{HyperLogLog: the
  analysis of a near-optimal cardinality estimation algorithm}}, in: P.~Jacquet
  (Ed.), {Analysis of Algorithms 2007 (AofA07)}, Juan les pins, France, 2007,
  pp. 127--146.
\newline\urlprefix\url{https://hal.archives-ouvertes.fr/hal-00406166}

\bibitem{DC2009:OrderStatisticsEstimatingCardinalitiesMassiveDataSets}
F.~Giroire,
  \href{http://www.sciencedirect.com/science/article/pii/S0166218X08002813}{Order
  statistics and estimating cardinalities of massive data sets}, Discrete
  Applied Mathematics 157~(2) (2009) 406 -- 427.
\newblock \href {http://dx.doi.org/https://doi.org/10.1016/j.dam.2008.06.020}
  {\path{doi:https://doi.org/10.1016/j.dam.2008.06.020}}.
\newline\urlprefix\url{http://www.sciencedirect.com/science/article/pii/S0166218X08002813}

\bibitem{DC:aLinearTimeProbabilisticCountingDatabaseApp}
K.-Y. Whang, B.~T. Vander-Zanden, H.~M. Taylor,
  \href{http://doi.acm.org/10.1145/78922.78925}{A linear-time probabilistic
  counting algorithm for database applications}, ACM Trans. Database Syst.
  15~(2) (1990) 208--229.
\newblock \href {http://dx.doi.org/10.1145/78922.78925}
  {\path{doi:10.1145/78922.78925}}.
\newline\urlprefix\url{http://doi.acm.org/10.1145/78922.78925}

\bibitem{hash_UniversalClassesOfHashFunctions}
J.~Carter, M.~N. Wegman,
  \href{http://www.sciencedirect.com/science/article/pii/0022000079900448}{Universal
  classes of hash functions}, Journal of Computer and System Sciences 18~(2)
  (1979) 143 -- 154.
\newblock \href
  {http://dx.doi.org/http://dx.doi.org/10.1016/0022-0000(79)90044-8}
  {\path{doi:http://dx.doi.org/10.1016/0022-0000(79)90044-8}}.
\newline\urlprefix\url{http://www.sciencedirect.com/science/article/pii/0022000079900448}

\bibitem{TON2017_CardinalityEstimationElephantFlowsACompactSolutionBasedVirtualRegisterSharing}
Q.~Xiao, S.~Chen, Y.~Zhou, M.~Chen, J.~Luo, T.~Li, Y.~Ling, Cardinality
  estimation for elephant flows: A compact solution based on virtual register
  sharing, IEEE/ACM Transactions on Networking 25~(6) (2017) 3738--3752.
\newblock \href {http://dx.doi.org/10.1109/TNET.2017.2753842}
  {\path{doi:10.1109/TNET.2017.2753842}}.

\bibitem{SDC2010SHLL:SlidingHyperLogLogEstimatingCardinalityDataStreamOverSlidingWindow}
Y.~Chabchoub, G.~Hebrail, Sliding hyperloglog: Estimating cardinality in a data
  stream over a sliding window, in: 2010 IEEE International Conference on Data
  Mining Workshops, 2010, pp. 1297--1303.
\newblock \href {http://dx.doi.org/10.1109/ICDMW.2010.18}
  {\path{doi:10.1109/ICDMW.2010.18}}.

\bibitem{SDC2007:EstimatingNumberActiveFlowsDataStreamOverSlidingWindow}
E.~Fusy, F.~Giroire,
  \href{http://dl.acm.org/citation.cfm?id=2791135.2791142}{Estimating the
  number of active flows in a data stream over a sliding window}, in:
  Proceedings of the Meeting on Analytic Algorithmics and Combinatorics, ANALCO
  '07, Society for Industrial and Applied Mathematics, Philadelphia, PA, USA,
  2007, pp. 223--231.
\newline\urlprefix\url{http://dl.acm.org/citation.cfm?id=2791135.2791142}

\bibitem{GLOBECOM2003_CountingNetworkFlowsInRealTime}
H.-A. Kim, D.~R. O'Hallaron, Counting network flows in real time, in: Global
  Telecommunications Conference, 2003. GLOBECOM '03. IEEE, Vol.~7, 2003, pp.
  3888--3893 vol.7.
\newblock \href {http://dx.doi.org/10.1109/GLOCOM.2003.1258959}
  {\path{doi:10.1109/GLOCOM.2003.1258959}}.

\bibitem{NETWORKING2009:CountingFlowsOverSlidingWindowsInHighSpeedNetworks}
J.~Sanju{\`a}s-Cuxart, P.~Barlet-Ros, J.~Sol{\'e}-Pareta, Counting flows over
  sliding windows in high speed networks, in: L.~Fratta, H.~Schulzrinne,
  Y.~Takahashi, O.~Spaniol (Eds.), NETWORKING 2009, Springer Berlin Heidelberg,
  Berlin, Heidelberg, 2009, pp. 79--91.

\bibitem{Iwqos2017:ACE:PerflowCountingForBigNetworkDataStreamOverSlidingWindows}
Y.~Zhou, Y.~Zhou, S.~Chen, Y.~Zhang, Per-flow counting for big network data
  stream over sliding windows, in: 2017 IEEE/ACM 25th International Symposium
  on Quality of Service (IWQoS), 2017, pp. 1--10.
\newblock \href {http://dx.doi.org/10.1109/IWQoS.2017.7969118}
  {\path{doi:10.1109/IWQoS.2017.7969118}}.

\bibitem{HSD2016_CVSFastCardinalityEstimationForLargeScaleDataStreamsOverSlidingWindows}
J.~Shan, J.~Luo, G.~Ni, Z.~Wu, W.~Duan,
  \href{http://www.sciencedirect.com/science/article/pii/S0925231216002320}{Cvs:
  Fast cardinality estimation for large-scale data streams over sliding
  windows}, Neurocomputing 194 (2016) 107 -- 116.
\newblock \href
  {http://dx.doi.org/https://doi.org/10.1016/j.neucom.2016.01.072}
  {\path{doi:https://doi.org/10.1016/j.neucom.2016.01.072}}.
\newline\urlprefix\url{http://www.sciencedirect.com/science/article/pii/S0925231216002320}

\bibitem{SDC2017:FastCountingCardinalityOfFlowsBigTrafficOverSlidingWindows}
J.~Shan, Y.~Fu, G.~Ni, J.~Luo, Z.~Wu,
  \href{https://doi.org/10.1007/s11704-016-6053-x}{Fast counting the
  cardinality of flows for big traffic over sliding windows}, Frontiers of
  Computer Science 11~(1) (2017) 119--129.
\newblock \href {http://dx.doi.org/10.1007/s11704-016-6053-x}
  {\path{doi:10.1007/s11704-016-6053-x}}.
\newline\urlprefix\url{https://doi.org/10.1007/s11704-016-6053-x}

\bibitem{HSD:ADataStreamingMethodMonitorHostConnectionDegreeHighSpeed}
P.~Wang, X.~Guan, T.~Qin, Q.~Huang, A data streaming method for monitoring host
  connection degrees of high-speed links, IEEE Transactions on Information
  Forensics and Security 6~(3) (2011) 1086--1098.
\newblock \href {http://dx.doi.org/10.1109/TIFS.2011.2123094}
  {\path{doi:10.1109/TIFS.2011.2123094}}.

\bibitem{HSD:DetectionSuperpointsVectorBloomFilter}
W.~Liu, W.~Qu, J.~Gong, K.~Li, Detection of superpoints using a vector bloom
  filter, IEEE Transactions on Information Forensics and Security 11~(3) (2016)
  514--527.
\newblock \href {http://dx.doi.org/10.1109/TIFS.2015.2503269}
  {\path{doi:10.1109/TIFS.2015.2503269}}.

\bibitem{HSD:SpreaderClassificationBasedOnOptimalDynamicBitSharing}
T.~Li, S.~Chen, W.~Luo, M.~Zhang, Y.~Qiao, Spreader classification based on
  optimal dynamic bit sharing, IEEE/ACM Transactions on Networking 21~(3)
  (2013) 817--830.
\newblock \href {http://dx.doi.org/10.1109/TNET.2012.2218255}
  {\path{doi:10.1109/TNET.2012.2218255}}.

\bibitem{HSD:GPU:2014:AGrandSpreadEstimatorUsingGPU}
S.-H. Shin, E.-J. Im, M.~Yoon,
  \href{http://www.sciencedirect.com/science/article/pii/S0743731513002189}{A
  grand spread estimator using a graphics processing unit}, Journal of Parallel
  and Distributed Computing 74~(2) (2014) 2039 -- 2047.
\newblock \href
  {http://dx.doi.org/http://dx.doi.org/10.1016/j.jpdc.2013.10.007}
  {\path{doi:http://dx.doi.org/10.1016/j.jpdc.2013.10.007}}.
\newline\urlprefix\url{http://www.sciencedirect.com/science/article/pii/S0743731513002189}

\bibitem{PD2013:BenchmarkingOfCommunicationTechniquesForGPUs}
M.~Bernaschi, M.~Bisson, D.~Rossetti,
  \href{http://www.sciencedirect.com/science/article/pii/S0743731512002213}{Benchmarking
  of communication techniques for gpus}, Journal of Parallel and Distributed
  Computing 73~(2) (2013) 250 -- 255.
\newblock \href {http://dx.doi.org/https://doi.org/10.1016/j.jpdc.2012.09.006}
  {\path{doi:https://doi.org/10.1016/j.jpdc.2012.09.006}}.
\newline\urlprefix\url{http://www.sciencedirect.com/science/article/pii/S0743731512002213}

\bibitem{PD2013:GeneratingDataTransfersForDistributedGPUParallelPrograms}
F.~Silber-Chaussumier, A.~Muller, R.~Habel,
  \href{http://www.sciencedirect.com/science/article/pii/S0743731513001603}{Generating
  data transfers for distributed gpu parallel programs}, Journal of Parallel
  and Distributed Computing 73~(12) (2013) 1649 -- 1660, heterogeneity in
  Parallel and Distributed Computing.
\newblock \href {http://dx.doi.org/https://doi.org/10.1016/j.jpdc.2013.07.022}
  {\path{doi:https://doi.org/10.1016/j.jpdc.2013.07.022}}.
\newline\urlprefix\url{http://www.sciencedirect.com/science/article/pii/S0743731513001603}

\bibitem{expdata:Caida}
C.~for Applied Internet Data~Analysis, The caida anonymized internet traces,
  \url{http://www.caida.org/data/passive}, online;accessed 2017 (2017).

\bibitem{expdata:IPtraceCernetJS}
CERNET, China education and research network,
  \url{http://iptas.edu.cn/src/system.php}, online;accessed 2017 (2017).

\end{thebibliography}

\end{document}